\documentclass[11pt]{amsart}
\usepackage{amsmath}

\usepackage{amsfonts,amssymb}
\usepackage{graphicx}
\usepackage{enumerate}
\usepackage{amsthm}
\usepackage[all]{xy}
\usepackage{extarrows}
\usepackage{color}
\usepackage{lmodern}
\usepackage{shuffle}
\usepackage{mathabx}
\usepackage{lmodern}
\usepackage{tikz}
\usepackage{hyperref}
\usepackage{rotating} 
\usepackage{scalefnt}

\usepackage{xspace}
\usepackage{axodraw4j}

\usepackage{a4wide}

\def\N{\mathbb N}

\def\1{\mathbf 1}

\newtheorem{theorem}{Theorem}[section]
\newtheorem {lemma}[theorem]{Lemma}
\newtheorem {proposition}[theorem]{Proposition}
\newtheorem {definition}[theorem]{Definition}
\newtheorem {corollary}[theorem]{Corollary}

\newtheorem {example}[theorem]{Example}
\newtheorem {remark}[theorem]{Remark}

\newcommand{\nc}{\newcommand}

\nc{\smop}[1]{\mathop{\hbox {\eightrm #1} }\nolimits}
\nc{\un}{{\mathbf 1}}
\nc{\id}{{\mathrm{id}}}

\include{graphs-Paper}

\begin{document}

%%%%%%%%%%%%%%%%%%%%%%%%%%%%%%%%

\title[Effective Action in Free Probability]{Effective Action in Free Probability}

%%%%%%%%%%%%%%%%%%%%%%%%%%%%%%%%

\author[K.~Ebrahimi-Fard]{Kurusch~Ebrahimi-Fard}
\address{Department of Mathematical Sciences, NTNU, NO 7491 Trondheim, Norway.}
\email{kurusch.ebrahimi-fard@ntnu.no}
\urladdr{https://folk.ntnu.no/kurusche/}

\author[F.~Patras]{Fr\'ed\'eric~Patras}
\address{Univ.~C\^ote d'Azur, CNRS, UMR 7351, Parc Valrose, 06108 Nice Cedex 02, France.}
\email{patras@unice.fr}
\urladdr{www-math.unice.fr/$\sim$patras}

\date{\today}

\maketitle

%%%%%%%%%%%%%%%%%%%%%%%%%%%%%%%%

\begin{abstract}
Recent works have explored relations between classical and quantum statistical physics on the one hand and Voiculescu's theory of free probability on the other. Motivated by these results, the present work focuses on the notion of effective action, which is closely related to the large deviation rate function in classical probability and one-particle irreducible correlation functions in quantum field theories. The central aim is to understand how it can be defined and studied in free probability.  In this respect, we introduce a suitable diagrammatic formalism.
\end{abstract}

%%%%%%%%%%%%%%%%%%%%%%%%%%%%%%%%

%%%%%%%%%%%%%%%%%%%%%%%%%%%%%%%%

\section{Introduction}
\label{sec:intro}

Recent works \cite{Bauer,Hruza} brought to light unexpected relations between large deviations for exclusion processes in classical as well as quantum statistical physics and Voiculescu's free probability. Motivated by these findings along with other, more classical, results and problems (see below), we investigate in the present article how the notion of effective action (the Legendre transform of the free energy or connected Green's functions), which is closely related to the large deviation rate function in classical probability and to one-particle irreducible (1PI) correlation (or Green's) functions in field theories, can be defined and studied in the context of free probability.  Our approach is algebraic and motivates the introduction of a set of diagrammatical rules suitable for representing relations between non-commutative generating functions.

The work at hand builds on two ideas. First, effective action --- 1PI Green's function, in the language of quantum or statistical field theories --- can be accounted for algebraically. This phenomenon was investigated using Hopf algebra techniques by one of us together with Ch.~Brouder \cite{bp} in the classical framework of  interacting quantum fields. Here, we would like to extend the idea to Voiculescu's theory of free probability \cite{voiculescu1,voiculescu2}. This requires novel tools due to constraints imposed by the notion of freeness that cannot be accounted for using the Hopf algebraic approach introduced in \cite{bp}. 

Second, planar quantum field theory (QFT) and the related notion of master field equations provide a setting for the perturbative description of the so-called planar sector of quantum field theories such as, for instance, quantum chromodynamics (QCD) and the behaviour of large-$N$ matrix models. It can be traced back to the 1974 work of 'tHooft on strong interactions \cite{Hooft}.

In the planar sector of QFTs, the computation of symmetry factors changes (and simplifies) compared to  usual quantum and statistical field theories. In fact, planarity is reflected in a strictly non-commutative nature of the theory, which results in a rather substantial deviation from the classical picture of the relations between different types of Green's functions. Indeed, it turns out that the description of these relations, as presented by Cvitanovich et al.~in \cite{cvitanovic1,cvitanovic2}, is closely related to Speicher's formulation of the relations between moments and free cumulants \cite{biane,nicaspeicher,novaksniady,speicher} in the context of Voiculescu's free probability. 
The connection between free probability and planar QFT has been explored in several works, in relation to master field equations, see e.g.~\cite{douglas,gopagross}. The latter reference points out for example (see p.~391) that ``the master field function $M(z)$ is such that $zM(z)$ is the generating functional of connected Green's functions''. In the univariate case, the master field would then identify with the ``conjugate field'', which we will introduce to perform Legendre transformations and to define effective action in free probability. This is one of the phenomena we came across that certainly deserves further investigations.

As for the planar graph concept of one-particle irreducibility (corresponding to connected Feynman diagrams that remain connected if one edge is removed), it does not seem to have generated much attention. In general, works on planar one-particle irreducibility are related to aspects of the planar sector of a QFT with $\phi^3$ and/or $\phi^4$ interactions, or, --- as in 'tHooft's original work --- to QCD. As a corollary of this fact, it follows that perturbative expansions in these works are not always suited to be transferred to the realm of free probability, where perturbative expansions of Green's/correlation functions in terms of Feynman diagrams, characteristic of quantum or statistical field theories, do not appear in general, and where the very notion of one-particle irreducibility  does not therefore immediately make sense. In fact, we shall define the free probability analog of 1PI Green's functions without ever appealing to the graphical notion of one-particle irreducibility.

Concretely, we recall that planar 1PI Green's functions appeared in the 1977 work \cite{Koplik} by Koplik, Neveu and Nussinov, which explored combinatorial and asymptotic properties of the planar sector of $\phi^3$- and $\phi^4$-theories. An extremely influential work in the general context of topological expansions was Br\'ezin's et al.~ landmark 1978 article {\it Planar Diagrams} \cite{Brezin}. However, one may remark that in there the treatment of one-particle irreducibility is brief and mostly limited to examples (related to $\phi^4$ interaction, see e.g.~p.~43). For the purposes of the present article, the main references, and the ones we will dwell on, are later works by Cvitanovich et al.~\cite{cvitanovic1,cvitanovic2}, where the relations between various planar Green's functions were studied graphically as well as from the generating series point of view. We should point out that the graphical calculus we will develop directly applies to planar QFT and somehow systematises Cvitanovich's et al.~presentation.

In free probability, the notions of effective action and one-particle irreducibility have, to the best of our knowledge, not been addressed specifically yet. The related references we are aware of point indeed in another direction than the present article. We refer to \cite[Chap. 7]{Mingo} and references therein. The idea there is to investigate large deviation phenomena for random matrices, leading in particular to a definition of free entropy and large deviation rate functions (as in the Ben Arous--Guionnet Theorem \cite[Chap.~7, Thm.~3]{Mingo}), mostly by analytic means. The problem of studying large deviations in this context has however direct connections with the present article. Firstly, it relates to our initial motivation, i.e., recent results connecting large deviations (but for exclusion processes) to free probability. Secondly, as emphasised in \cite[Sect.~7.2]{Mingo}, in the classical case, large deviations phenomena are accounted for by a large deviation rate function obtained by Legendre transform from cumulant generating series --- we will similarly define the free effective action by a suitable non-commutative Legendre transform of the free cumulant generating series.

\smallskip

The aim of the present article is ultimately twofold: first, it provides a mathematically rigorous definition of the notion of effective action in the context of the theory of free probability, by drawing analogy with the definition of 1PI Green's functions in planar QFT; second, we develop a set of diagrammatical rules on planar trees with decorated leaves and vertices suitable for a  graphical representation of the recursive construction of the terms of the formal power series expansion of the effective action.

Lastly, we should mention that in a previous work \cite{EP3}, we explained how various algebraic structures that we had introduced in free probability \cite{EP1,EP2} can be used to revisit planar QFT. In the preprint version of our article\footnote{\href{https://arxiv.org/abs/1509.08678}{arXiv:1509.08678v1}}, we investigated in particular planar connected Green's functions --- corresponding to free cumulants in free probability --- but also planar 1PI Green's functions. However, the latter part was preliminary and we decided to publish only the part relevant to the relation between planar full and connected Green's functions. The present article is an update to the approach to planar 1PI Green's functions as presented in \cite{EP3}, including re-elaboration and upgrading the materials contained in the aforementioned preprint. Graphical representations have been in particular improved.

\medskip

The paper is organised as follows. For the sake of completeness, we briefly recall in the next section fundamental notions in free probability and develop foundations of planar functional calculus, extending formulas to be found in Cvitanovich et al.~\cite{cvitanovic1,cvitanovic2} to the algebraic approach to non-commutative probability theory developed in reference \cite{EPTZ}. We introduce then the effective action --- by Legendre transform --- together with elements and formulas of planar functional calculus adapted to the study of effective action. The last section develops a diagrammatical formulation of the recursive construction of the terms in the formal power series expansion of the effective actions. The relevant diagrams are disks inside which certain planar trees are drawn, whose leaves are marked points of the boundary. These marked leaves are totally ordered anti-clockwise. The concept of characteristic path in these trees plays a specific role. It is the extremal path joining the first and last leaves. Higher order diagrams, corresponding to components in the expansion of higher order terms of the effective action can be constructed inductively from local operations on the vertices and edges of this path.

\medskip

{\bf{Notation}}: The set of {\it strictly} positive integers is denoted $\N$. All algebraic structures are defined over a commutative base field  $\mathbb K$ of characteristic zero, with unit denoted $1_{\mathbb K}$ (or simply $1$ when no confusion can arise). If not stated otherwise, Einstein's summation convention is in place, i.e., repeated indices imply summation. For example, for an index $i$ varying in $[n]:=\{1,\ldots,n\}$, $a_iX_i:=\sum_{i=1}^n a_iX_i$. We let ${\mathbf X}^\ast$ denote the set of {\it non-empty} sequences, called words, over the alphabet ${\mathbf X}=\{x_1,\ldots, x_n\}$. The empty word $\emptyset$ makes ${\mathbf X}^\ast \cup\{\emptyset\}$ the free monoid defined over ${\mathbf X}$. As usual, $\delta_{i,j}$ stands for the Kronecker symbol (which equals $1$ if $i=j$ and $0$ else).

\medskip

{\bf{Acknowledgements}}: KEF is supported by the Research Council of Norway through
project 302831 “Computational Dynamics and Stochastics on Manifolds” (CODYSMA). FP acknowledges support from the ANR project Algebraic Combinatorics, Renormalization, Free probability and Operads – CARPLO (Project-ANR-20-CE40-0007) and from the ANR -- FWF project PAGCAP.

%%%%%%%%%%%%%%%%%%
%%%%%%%%%%%%%%%%%%

\section{Free Probability and free cumulants}
\label{sec:freeproba}

Let $(A,\varphi)$ be a non-commutative probability space \cite{nicaspeicher}, that is, it consists of an associative $\mathbb K$-algebra $A$ with unit $1_A$ and a unital linear form $\varphi$ ($\varphi(1_A)=1$). Elements of $A$ are called non-commutative random variables. Standard examples of such spaces are algebras of random matrices, with $\varphi$ being the expectation of the trace (normalised by $1/N$ for $N \times N$ matrices, so that $\varphi$ is unital), and algebras of observables in quantum mechanics. In the latter case the map $\varphi$ is then the average value of observables defined over the ground state of the quantum system. 

To a family $a_1,\ldots,a_n$ of elements of $A$ is then associated the generating series of multivariate moments and free cumulants,
\begin{equation}
\label{eq:moments}
	\mathrm{M}({\mathbf x})
	:= 1 + \sum_{q>0}\sum_{0< i_1,\ldots,i_q\leq n} 
	\mathrm{m}^{(q)}_{i_1 \cdots i_q} x_{i_1} \cdots x_{i_q}
\end{equation}
respectively
\begin{equation}
\label{eq:cumulants}
	\mathrm{K}({\mathbf x})
	:= \sum_{p>0} \sum_{0 < i_1,\ldots,i_p \leq n}
	\mathrm{k}^{(p)}_{i_1 \cdots i_p} x_{i_1} \cdots x_{i_p},
\end{equation}
where $\mathrm{m}^{(q)}_{i_1 \cdots i_q}:=\varphi(a_{i_1}\cdots a_{i_q}) \in \mathbb K$ --- the cumulants $\mathrm{k}^{(p)}_{i_1 \cdots i_p} \in\mathbb K$ are defined below. Here the formal variables $\{x_i\}_{1\leq i\leq n}$ do not commute, so that $\mathrm{M}=\mathrm{M}({\mathbf x})$ and $\mathrm{K}=\mathrm{K}({\mathbf x})$ are non-commutative formal power series. In particular the coefficients $\mathrm{m}^{(q)}_{i_1 \cdots i_q}$ and $\mathrm{k}^{(p)}_{i_1 \cdots i_p}$ are, in general, not invariant by permutation of the indices.

The generating series of moments \eqref{eq:moments} and free cumulants \eqref{eq:cumulants} are related by the functional equation \cite{nicaspeicher,speicher}: 
\begin{equation}
\label{eq:speicher}
	\mathrm{M}({\mathbf x}) 
	=1 + \mathrm{K}({\mathbf x}\mathrm{M}({\mathbf x})),
\end{equation}
where the term $\mathrm{K}({\mathbf x}\mathrm{M}({\mathbf x}))$ on the righthand side means that $x_i\mathrm{M}({\mathbf x})$ is substituted for the letter $x_i$ in the formal power series expansion of $\mathrm{K}({\mathbf x})$. This formula can be considered as a definition of free cumulants $\mathrm{k}^{(p)}_{i_1 \cdots i_p}$. In low degrees, we find the well-known relations \cite{Mingo,nicaspeicher}
\allowdisplaybreaks{
\begin{align}
	\mathrm{m}^{(0)} 
	&= 1  											\nonumber \\
	\mathrm{m}_{j_1}^{(1)} 
	&= \mathrm{k}_{j_1}^{(1)} 							\nonumber \\ 
	\mathrm{m}_{j_1j_2}^{(2)} 
	&= \mathrm{k}_{j_1}^{(1)}\mathrm{m}_{j_2}^{(1)} 
					+ \mathrm{k}_{j_1j_2}^{(2)}\mathrm{m}^{(0)}
			= \mathrm{k}_{j_1}^{(1)}\mathrm{k}_{j_2}^{(1)} 
					+ \mathrm{k}_{j_1j_2}^{(2)}			\nonumber \\ 
	\mathrm{m}_{j_1j_2j_3}^{(3)} 
	&= 	 {\mathrm{k}_{j_1}^{(1)}}\mathrm{m}_{j_2j_3}^{(2)}
						+  {\mathrm{k}_{j_1j_2}^{(2)}}\mathrm{m}_{j_3}^{(1)}
						+  {\mathrm{k}_{j_1j_3}^{(2)}}\mathrm{m}_{j_2}^{(1)}
						+  \mathrm{k}_{j_1j_2j_3}^{(3)}	\mathrm{m}^{(0)} 		\nonumber \\
	&= {\mathrm{k}_{j_1}^{(1)}}{\mathrm{k}_{j_2}^{(1)}}{\mathrm{k}_{j_3}^{(1)}} 
						+ \mathrm{k}_{j_1}^{(1)} \mathrm{k}_{j_2j_3}^{(2)}
						+ \mathrm{k}_{j_1j_2}^{(2)} \mathrm{k}_{j_3}^{(1)} 
						+ \mathrm{k}_{j_1j_3}^{(2)} \mathrm{k}_{j_2}^{(1)}  
						+ \mathrm{k}_{j_1j_2j_3}^{(3)}						\nonumber \\
	\mathrm{m}_{j_1j_2j_3j_4}^{(4)} 
	&= {\mathrm{k}_{j_1}^{(1)}}\mathrm{m}_{j_2j_3j_4}^{(3)} 
				+ {\mathrm{k}_{j_1j_2}^{(2)}}\mathrm{m}_{j_3j_4}^{(2)} 
				+ {\mathrm{k}_{j_1j_3}^{(2)}}\mathrm{m}_{j_2}^{(1)}\mathrm{m}_{j_4}^{(1)} 
				+ {\mathrm{k}_{j_1j_4}^{(2)}}\mathrm{m}_{j_2j_3}^{(2)} 			\nonumber \\
	&   \quad + {\mathrm{k}_{j_1j_2j_3}^{(3)}}\mathrm{m}_{j_4}^{(1)} 
				+ {\mathrm{k}_{j_1j_2j_4}^{(3)}}\mathrm{m}_{j_3}^{(1)} 
				+ {\mathrm{k}_{j_1j_3j_4}^{(3)}}\mathrm{m}_{j_2}^{(1)}  
				+ \mathrm{k}_{j_1j_2j_3j_4}^{(4)}\mathrm{m}^{(0)}. 				\nonumber 
\end{align}}

%%%%%%%%%%%%%%%%%%
%%%%%%%%%%%%%%%%%%

\section{Planar algebraic calculus}
\label{sec:series}

The constructions in this section aim at describing the proper algebraic framework to deal with formulas like those relating moments and free cumulants \eqref{eq:speicher} as well as with more advanced formulas such as the ones we will encounter later on. They include (up to some minor variants) and complement those in reference \cite{EPTZ}. We remark that the section title ``free algebraic calculus'' would have been more meaningful than ``planar algebraic calculus'', but it also would have likely created ambiguities and confusion from the viewpoint of established mathematical terminology. 

The planar algebraic formalism appears to govern free probability but also Boolean and monotone probability, providing therefore a joint framework for the three theories. In fact, they appear more and more to be entangled, so that various advanced computations in free probability appeal to ideas and formulas in the other two theories --- see \cite{EPTZ} for explanations and references.

Given the alphabet ${\mathbf X}=\{x_1,x_2,x_3,\ldots, x_n\}$, we denote by $R_{\mathbf X}$ the ring of non-commutative formal power series with coefficients in the base field ${\mathbb K}$
$$
	R_{\mathbf X} \coloneq {\mathbb K} \langle\langle x_1,x_2,x_3,\ldots , x_n\rangle\rangle.
$$ 

An element $f=f({\mathbf x}) \in R_{\mathbf X}$ can be written
\[
	f({\mathbf x}) = 
	f_0 + \sum_{k=1}^\infty 
	\sum_{(i_1,\dotsc,i_k)\in [n]^k}f_{i_1\dotsm i_k} x_{i_1} \dotsm x_{i_k} 
	= \sum_{w \in [n]^* \cup \{\emptyset\}} f_{w} x_{w}
\]
with coefficients $f_0, f_{i_1\dotsm i_k} \in \mathbb K$ and where we have associated to a word $w=i_1 \dotsm i_m \in {[n]}^m$ the non-commutative monomial $x_w \coloneq x_{i_1}\dotsm x_{i_m}$, with the convention that $x_\emptyset =1$ and $f_\emptyset=f_0$. As in the previous section, we will also write $f_{i_1\dotsm i_k}^{(k)}$ for $ f_{i_1\dotsm i_k}$.

An element $f=f({\mathbf x}) \in R_{\mathbf X}$ is called a {\it functional} as it can be viewed alternatively as a function on ${\mathbf X}^\ast\cup\{\mathbf 1\}$. A functional is {\it finite} if it is with finite support when seen as a function on ${\mathbf X}^\ast\cup\{\mathbf 1\}$.

We then set 
$$
	G^1_{\mathbf X}:=\{f({\mathbf x}) \in R_{\mathbf X} \ |\  f_0=1\}, \qquad 
	G^0_{\mathbf X}:=\{f({\mathbf x}) \in R_{\mathbf X} \ |\  f_0=0\},
$$
so that $G^1_{\mathbf X}=1 + G^0_{\mathbf X}$. A family 
$$
	{\mathbf g}=(g_1,\ldots,g_n) \in R_{\mathbf X}^n
$$ 
is called a {\it field} --- we shall encounter naturally such families later on.

The set of so-called ``tangent-to-identity'' elements is the subset of $G^0_{\mathbf X}$ defined by
$$
	G^c_{\mathbf X}:=\{f \in R_{\mathbf X}\ |\  f_0=0,\ f_i=1, \forall i \in [n] \}. 
$$

The usual product of non-commutative formal power series (also known as Cauchy product) is denoted $fg$ for elements $f,g \in R_{\mathbf X}$, with coefficient corresponding to the word $w=i_1\dotsm i_m$ defined as
\[
	(fg)_{i_1\dotsm i_m} 
	:= f_{i_1\dotsm i_m}g_0 + f_0g_{i_1\dotsm i_m} 
	+\sum_{j=1}^{m-1} f_{i_1 \dotsm i_j} g_{i_{j+1} \dotsm i_m}. 
\]
By standard arguments, $G^1_{\mathbf X}$ is a group under this product. 

The composition $f\circ {\mathbf g}$ of a functional $f \in R_{\mathbf X}$ and a field without constant terms, ${\mathbf g}=(g_1,\dots,g_n)\in (G^0_{\mathbf X})^n$, is obtained by substitution of $g_i$ for $x_i$, that is:
$$
	f\circ {\mathbf g}({\mathbf x}) = 
	f_0 + \sum_{k=1}^\infty \sum_{(i_1,\dotsc,i_k)\in [n]^k}f_{i_1\dotsm i_k} 
	g_{i_1}({\mathbf x})  \dotsm g_{i_k}({\mathbf x}) .
$$
If the functional is finite, then the hypothesis that the field must be without constant terms can be dropped. A similar observation applies to composition of fields below. The composition of a field ${\mathbf f}=(f_1,\ldots,f_n)\in R_{\mathbf X}^n$ with a field without constant terms, ${\mathbf g}=(g_1,\ldots,g_n)\in (G^0_{\mathbf X})^n$, is defined term-wise as $(f_1\circ {\mathbf g},\ldots ,f_n\circ {\mathbf g})$.

\begin{definition}
For $f,g \in R_{\mathbf X}$, {{\it shifted composition}} is defined by:
\begin{equation}
\label{monoprod}
	 (f \bullet g)({\mathbf x}) := g({\mathbf x})f({\mathbf x}g({\mathbf x})),
\end{equation}
where 
$$
	f({\mathbf x}g({\mathbf x}))
	:=f_0 + \sum_{k=1}^\infty \sum_{(i_1,\dotsc,i_k)\in [n]^k}f_{i_1\dotsm i_k} x_{i_1}
	g({\mathbf x}) \dotsm x_{i_k}g({\mathbf x}).
$$
\end{definition}

In the case of a single-letter (${\mathbf X}=\{x\}=\{x_1\}$), one has a linear isomorphism $G^1_{\mathbf X} \cong G^c_{\mathbf X}$ given by $\mu: f \longmapsto xf$ and
$$
	\mu(f \bullet g)(x)
	=xg(x)f(xg(x))
	=\mu(f)(\mu(g)(x)).
$$
By this isomorphism, $(G^1_{\mathbf X},\bullet)$ is isomorphic to the group of tangent-to-identity formal diffeomorphisms of the line. 

\begin{proposition}
The new product \eqref{monoprod} defined on $R_{\mathbf X}$ is associative. 
\end{proposition}

\begin{proof}
Indeed, a quick computation shows that 
\begin{align*}
  ((f \bullet g)  \bullet h)({\mathbf x}) 
  &=  h({\mathbf x}) (f\bullet g)({\mathbf x}{h({\mathbf x})})\\
  &= h({\mathbf x})g({\mathbf x}h({\mathbf x})) f({\mathbf x}{h({\mathbf x})}g({\mathbf x}h({\mathbf x})))\\
  &= (f \bullet (g  \bullet h))({\mathbf x}).
\end{align*}
\end{proof}
 
In particular (see \cite{EPTZ}), we have the following

\begin{proposition}
$G^\bullet:=(G^1_{\mathbf X},\bullet)$ is a non-commutative group with unit~$1$. 
\end{proposition}

Following our previous observation on the single-letter case, the group $G^\bullet$ can be thought of as a multidimensional generalisation of the group of tangent-to-identity formal diffeomorphisms. It is a left-linear group \cite[Sect.~6.4]{CP}, that is, the coefficients of the monomials $x_w$ in the formal power series expansion of the product $(f \bullet g)({\mathbf x})$ are polynomials in the coefficients of $f$ and $g$ that are linear in the former and multilinear in the  latter. 

It follows \cite[Chap.~6]{CP} that the tangent space $\mathfrak g=G^0$ at $1$ to $G^\bullet$ is a right pre-Lie algebra, that is, it is equipped with a bilinear product 
$$
	\triangleleft : \mathfrak g \times \mathfrak g \to \mathfrak g
$$
such that 
$$
	a\triangleleft(b\triangleleft c)-(a\triangleleft b)\triangleleft c
		=a\triangleleft(c\triangleleft b)-(a\triangleleft c)\triangleleft b.
$$
This relation implies that the bilinear product 
$$
	[a,b]_{\scriptscriptstyle{\triangleleft}} :=a\triangleleft b-b\triangleleft a
$$ 
is a Lie bracket.

The pre-Lie product $\triangleleft$ can be described explicitly by:
$$	
	x_{i_1}\cdots x_{i_n}\triangleleft x_{j_1}\cdots x_{j_m} 
	=\sum\limits_{k=0}^n x_{i_1}\cdots x_{i_k} x_{j_1}\cdots x_{j_m}x_{i_{k+1}}\cdots x_{i_n}.
$$
It plays an important role in non-commutative probability theory. For example, it allows to understand algebraically the relations between free, monotone and Boolean cumulants. See \cite{EPTZ} and references therein for details.

In the single-letter case, one finds 
$$
	x^n\triangleleft x^m=(n+1)x^{n+m},
$$
so that we get for the associated Lie bracket
$$
	[x^n,x^m]_{\scriptscriptstyle{\triangleleft}} =(n-m)x^{n+m}.
$$
This is, up to sign and isomorphism, the Lie algebra of polynomial vector fields $x^{n+1}\partial_x$.

\begin{definition}
For $f\in G^0_{\mathbf X}$, $g\in G^1_{\mathbf X}$, the product $f \bullet g $ is in $ G^0_{\mathbf X}$ and we set:
\begin{align}
\label{leftprod}
	 (f \prec g)({\mathbf x})
	 &:= f({\mathbf x}g({\mathbf x})),\\ 
\label{rightprod}
	 (f \succ g)({\mathbf x})
	 &:= (g({\mathbf x})-1)f({\mathbf x}g({\mathbf x})),
\end{align}
where $f\prec g$ and $f\succ g$ are in $G^0_{\mathbf X}$. Furthermore
$$
	f\bullet g=f\prec g+f\succ g.
$$
\end{definition}

\begin{example}
The identity $\mathrm{M}({\mathbf x})=1+\mathrm{K}({\mathbf x}\mathrm{M}({\mathbf x}))$ defining free cumulants in terms of moments can be rewritten:
$$
	\mathrm{M}({\mathbf x})=1+(\mathrm{K}\prec\mathrm{M})({\mathbf x}).
$$
\end{example}

\begin{lemma}
The following identity holds:
$$
	(f\prec g)\prec h=f\prec (g\bullet h).
$$
In other terms, the $\prec$-product defines a right module structure on $G^0_{\mathbf X}$ over the group $G^1_{\mathbf X}$.
\end{lemma}

\begin{proof}
Indeed, we check that
\begin{align*}
	((f\prec g)\prec h)({\mathbf x})
	&=(f\prec g)({\mathbf x}h({\mathbf x}))\\
	&=f({\mathbf x}h({\mathbf x})g({\mathbf x}h({\mathbf x})))\\
	&=(f\prec(g\bullet h))({\mathbf x}).
\end{align*}
\end{proof}

Finally, we have distributivity relations relating the Cauchy product with the other products: 
\begin{lemma}
for $f,g\in G^0_{\mathbf X}$ and $h\in G^1_{\mathbf X}$
\begin{align*}
	(fg)\bullet h
	&=(f\bullet h)(g\prec h)\\
	(fg)\prec h
	&=(f\prec h)(g\prec h)\\
	(fg)\succ h
	&=(f\succ h)(g\prec h).
\end{align*}
\end{lemma}

\begin{proof}
Indeed, we verify first
\begin{align*}
	((fg)\bullet h)({\mathbf x})
	&=h({\mathbf x})(fg)({\mathbf x}h({\mathbf x}))\\
	&=h({\mathbf x})f({\mathbf x}h({\mathbf x}))g({\mathbf x}h({\mathbf x}))\\
	&=(f\bullet h)(g\prec h)({\mathbf x}),
\end{align*}
and then
\begin{align*}
	((fg)\prec h)({\mathbf x})
	&=(fg)({\mathbf x}h({\mathbf x}))\\
	&=f({\mathbf x}h({\mathbf x}))g({\mathbf x}h({\mathbf x}))\\
	&= (f\prec h)(g\prec h)({\mathbf x}),
\end{align*}
and ultimately
\begin{align*}
	((fg)\succ h)({\mathbf x})
	&=(h({\mathbf x})-1)(fg)({\mathbf x}h({\mathbf x}))\\
	&=(h({\mathbf x})-1)f({\mathbf x}h({\mathbf x}))g({\mathbf x}h({\mathbf x}))\\
	&=(f\succ h)(g\prec h)({\mathbf x}).
\end{align*}
\end{proof}

%%%%%%%%%%%%%%%%%%
%%%%%%%%%%%%%%%%%%

\section{Planar functional calculus}
\label{sect:planarCalc}	

We recall here briefly the rules of planar (that is, non-commutative) functional calculus following Cvitanovic et al.~\cite{cvitanovic1,cvitanovic2}. Consider $f \in R_{\mathbf X}$,
$$
	f({\mathbf x})
	= f_0 + \sum_{k>0 \atop 0<i_1,\ldots,i_k\leq n} 
	f^{(k)}_{{i_1} \cdots {i_k}} x_{i_1} \cdots x_{i_k} ,
$$
so that
$$
	f^{(k)}_{{i_1} \cdots {i_k}} 	
	= \left. \frac{\partial^k}{\partial{x_{i_k}} \cdots \partial{x_{i_1}}}\right |_{{\mathbf x}=0}
	f({\mathbf x}),
$$
where the rules stated for functional derivations are
$$
	\frac{\partial}{\partial{x_j}} (u x_{i_1} \cdots x_{i_k} ) 
	:= u \delta_{ji_1} x_{i_2} \cdots x_{i_k}, 
$$ 
for $u \in \mathbb{K}$ a scalar --- in particular $\frac{\partial}{\partial{x_i}} u = 0$. Notice that Schwarz' rules of calculus do not apply, i.e., we have non-commutativity of derivations, $\frac{\partial^2}{\partial{x_1}\partial{x_2}}\not=\frac{\partial^2}{\partial{x_2}\partial{x_1}}$. All formulas in this section are direct consequences of these rules and we omit their proofs.

We introduce now the following notation:
$$
	f({\mathbf x})_{x_i}:=\frac{\partial}{\partial{x_i}}f({\mathbf x}),
	\qquad 
	f(0)_{x_i}:=f_i^{(1)},
$$
and call   
\begin{equation}
\label{difffield}
	\partial f:=(f({\mathbf x})_{x_1},\dots,f({\mathbf x})_{x_n})
\end{equation}
the {\it differential field} of $f$. More generally, we set:
$$
	f({\mathbf x})_{x_{i_1},\ldots,x_{i_k}}
	:=\frac{\partial^k}{\partial{x_{i_k}} \cdots \partial{x_{i_1}}}f({\mathbf x}),
	\qquad
	f(0)_{x_{i_1},\ldots,x_{i_k}}
	:=f_{{i_1} \cdots {i_k}}^{(k)}.
$$

The key ingredients in the planar functional calculus developed in \cite{cvitanovic1,cvitanovic2} are the planar Leibniz and chain rules. Let $f({\mathbf x}), g({\mathbf x}) \in R_{\mathbf X}$. Then the functional derivative of the Cauchy product $fg$ is given by the planar Leibniz rule
\begin{equation}
\label{Leibniz}
	\frac{\partial}{\partial{x_i}} (f({\mathbf x})g({\mathbf x})) 
	= \frac{\partial}{\partial{x_i}} (f({\mathbf x}))g({\mathbf x}) 
		+ f(0)\frac{\partial}{\partial{x_i}} g({\mathbf x}).  
\end{equation}
Here $f(0) = f_0 \in \mathbb{K}$. 

Now assume that $f ({\mathbf x}) \in R_{\mathbf X}$ is a functional and $\mathbf{g}:=(g_1,\ldots,g_n)\in (G^0_{\mathbf X})^n$ is a field without constant terms, we then get for the composition $f\circ\mathbf{g}$ the chain rule:
\begin{equation}
\label{chain}
	\frac{\partial}{\partial{x_i}}(f\circ \mathbf{g}) 
	= \frac{\partial g_m({\mathbf x}) }{\partial{x_i}} \cdot
								\frac{\partial}{\partial{g_m}}(f\circ \mathbf{g}).
\end{equation}
Recall that Einstein's summation convention is in place implying a sum over $m$ on the righthand side of \eqref{chain}. The last partial derivative is formally obtained from the expansion
$$
	f\circ \mathbf{g}
	=f_0 + \sum\limits_{k>0  \atop 0<i_1,\ldots,i_k\leq n }
	f^{(k)}_{i_1 \cdots i_k}g_{i_1}\cdots g_{i_k}.
$$

Dualising the construction of the differential field \eqref{difffield}, we define the {\it integral field} of $f\in R_{\mathbf X}$ by
\begin{equation}
\label{intfield}
	\int f:={\mathbf x}f({\mathbf x})=(x_1f({\mathbf x}),\ldots, x_nf({\mathbf x})).
\end{equation}
We get in particular that $f \circ \int g$ identifies with $f({\mathbf x}g({\mathbf x}))$ as previously defined. From the planar Leibniz and chain rules, we obtain the following rules for planar algebraic calculus:

\begin{lemma}
For $f\in G^0_{\mathbf X}$ and $g\in G^1_{\mathbf X}$,
\begin{align}
	\frac{\partial}{\partial{x_i}} (f\bullet g)({\mathbf x})
	&=\frac{\partial}{\partial{x_i}}g({\mathbf x})\cdot (f\prec g)({\mathbf x})
			+ (\frac{\partial}{\partial{x_i}}f)\bullet g({\mathbf x}),\label{chain1}\\
	\frac{\partial}{\partial{x_i}} (f\prec g)({\mathbf x})
	&=(\frac{\partial}{\partial{x_i}}f)\bullet g({\mathbf x}),\label{chain2}\\
	\frac{\partial}{\partial{x_i}} (f\succ g)({\mathbf x})
	&=\frac{\partial}{\partial{x_i}}g({\mathbf x})\cdot (f\prec g)({\mathbf x}). \label{chain3}
\end{align}
\end{lemma}
These rules are interesting in that they show that introducing functional derivatives implies the joint consideration of the three planar products, $\bullet, \prec,\succ$, although according our previous work each is naturally associated to one of the three fundamental non-commutative probability theories (respectively monotone, free and Boolean) \cite{EPTZ}. 

Notice also that the rules could be easily adapted to the more general case $f,g\in R_{\mathbf X}$.

\begin{corollary}\label{cordergenf}
Since $\mathrm{M}({\mathbf x})=1+(\mathrm{K}\prec \mathrm{M})({\mathbf x})$, we get for the moment generating function of a family of non-commutative random variables:
$$
	\frac{\partial}{\partial x_i}\mathrm{M}({\mathbf x})
		=(\frac{\partial }{\partial x_i}\mathrm{K})\bullet \mathrm{M}({\mathbf x}).
$$
\end{corollary}

%%%%%%%%%%%%%%%%%%
%%%%%%%%%%%%%%%%%%

\section{Change of variables and effective action}
\label{sec:changevar}

In this section we introduce the effective action in free probability through suitable changes of variables. This point of view is prompted by the very idea of non-commutative Legendre transform. We should emphasise that we were inspired by the works \cite{cvitanovic1,cvitanovic2} on planar QFT, in particular on planar one-particle irreducibility. Our aim is to extract a mathematically transparent and precise formulation, and to adapt it to the  setting of free probability.

\smallskip

We define now a formal change of variables  in terms of the integral field of a given $g \in G^1_{\mathbf X}$, by setting
\begin{equation}
\label{newvariables}
	{\mathbf Y}^g=(y^g_1,y^g_2,\ldots ,y^g_n):=\int g,
\end{equation}
where the definition of integral field \eqref{intfield} together with $g_0=1$ imply that
\begin{equation}
\label{ztransfrom1}
	y^g_i \coloneq x_ig(x) 
	=  \sum_{w \in [n]^* \cup \{ \emptyset \} } g_{w} x_{iw} 
	= x_i + \sum_{w \in [n]^*} g_{w} x_ix_w.
\end{equation}
For words $w=i_1 \cdots i_l \in [n]^*$, we set 
\begin{equation}
\label{ztransfrom2}
	y^g_{i_1 \cdots i_l} \coloneq y_{i_1}^g \cdots y_{i_l}^g.
\end{equation}

As the new variable $y_i^g$ is equal to $x_i$ plus quadratic and higher order terms in the letters $x_j$'s, any series $f({\mathbf x})$ in $R_{\mathbf X}$ can be rewritten uniquely as a formal power series in ${\mathbf Y}^g$, denoted $f^g=f^g({\mathbf y^g})$. In particular, one can write $x_i=z_i^g({\mathbf y}^g)$, where the righthand side of the equation denotes the expansion of $x_i$ as a formal power series in the $y_j^g$. 

In the sequel, it will often happen that the series $g({\mathbf x})$ will be fixed to be the moment series $\mathrm{M}({\mathbf x})$ defined in \eqref{eq:moments}, the generating series of moments associated to a family of elements in an algebra of non-commutative random variables. We will then simply write abusively ${\mathbf Y}$ for ${\mathbf Y}^g$ and $\mathbf y$ for $\mathbf y^g$.

%$f^g({\mathbf y})$ for $f^g({\mathbf y^g})$,  $f({\mathbf y})$ for $f({\mathbf y^g})$, and $f({\mathbf z})$ for $f({\mathbf z^g})$ (the series $f({\mathbf x})$ rewritten as a series in the variables $\mathbf y^g$, and equal as such to $f^g({\mathbf y^g})$).

Using this convention, the identity \eqref{eq:speicher} relating moments, $\mathrm{M}$, and free cumulants, $\mathrm{K}$, rewrites
$$	
	\mathrm{M}({\mathbf x})=1+\mathrm{K}({\mathbf y}),
$$
or, expressing the free cumulants in terms of moments,
$$
	\mathrm{K}({\mathbf y})=\mathrm{M}({\mathbf x})-1.
$$

%%%%%%%%%%%%%%%%%%

\subsection{Conjugate field}
\label{ssec:conjugate field}

Key to the definition of effective action are the rules of differential calculus relating expressions in the $x_i$'s and the $y_i$'s. Although a direct consequence of the chain rule, we state the following identity as 

\begin{proposition}
\label{difftrick}
With $g({\mathbf x})=\mathrm{M}({\mathbf x})$ in the change of variables \eqref{newvariables}, we have 
$$
	f({\mathbf x})=f^\mathrm{M}({\mathbf y})=(f^\mathrm{M}\prec \mathrm{M})({\mathbf x}),
$$ 
and therefore:
$$
	\frac{\partial}{\partial x_i}f({\mathbf x})
	=\frac{\partial y_k}{\partial x_i}\frac{\partial}{\partial y_k}f^\mathrm{M}({\mathbf y})
	=\mathrm{M}({\mathbf x})\frac{\partial}{\partial y_i}f^\mathrm{M}({\mathbf y})
	=(1+\mathrm{K}({\mathbf y}))\frac{\partial}{\partial y_i}f^\mathrm{M}({\mathbf y}),
$$
and we recover in particular the formula of Corollary \ref{cordergenf} expressed now in terms of $\mathbf y$ as:
$$
	\frac{\partial}{\partial x_i}\mathrm{M}({\mathbf x})
	=\frac{\partial y_k}{\partial x_i}\frac{\partial}{\partial y_k}\mathrm{K}({\mathbf y})
	=\mathrm{M}({\mathbf x})\frac{\partial}{\partial y_i}\mathrm{K}({\mathbf y})
	=(1+\mathrm{K}({\mathbf y}))\frac{\partial}{\partial y_i}\mathrm{K}({\mathbf y}).
$$
\end{proposition}

Let us move now to the Legendre transform of the cumulant generating series. Recall that in the simplest situation, the classical Legendre transform is obtained as follows. Let $f(t)$ be a strictly convex function. Its derivative $f'(t)$ is strictly increasing and can be inverted. Up to a constant, its Legendre transform $g(t)$ is then characterised by the identities $t=g'(f'(t))$ or $s=f'(g'(s))$. This is the definition we will adapt in the non-commutative multivariate context.

\begin{definition}[Conjugate field]
\label{def:conjfield}
Let ${\mathbf y}=(y_1,\dots, y_n)$ be the integral field of $\mathrm{M}$. Its conjugate field $\Phi=(\phi_1,\dots,\phi_n)$ is defined in terms of the free cumulants $\mathrm{K}({\mathbf y})$ by
\begin{equation}
\label{phi}
	\phi_i:=\frac{\partial \mathrm{K}({\mathbf y})}{\partial y_i}, \quad 1 \le i \le n.
\end{equation}
\end{definition}

Since $\mathrm{K}({\mathbf y})$ has no constant term, we have $\mathrm{K}({\mathbf y})=y_i\frac{\partial \mathrm{K}({\mathbf y})}{\partial y_i}$. Therefore we have 

\begin{proposition}
\label{prop:cumulant}
For the cumulant generating function $\mathrm{K}$ we have:
$$
	\mathrm{K}({\mathbf y})=y_i\phi_i.
$$
\end{proposition}

\begin{proposition}
\label{prop:moments}
The following equivalent identities hold true:
$$
	\frac{\partial}{\partial x_i}\mathrm{M}({\mathbf x})
	=\mathrm{M}({\mathbf x})\frac{\partial}{\partial y_i} \mathrm{K}({\mathbf y})
	=\mathrm{M}({\mathbf x})\phi_i,
$$
and
$$
	\mathrm{M}({\mathbf x})=1+x_i \mathrm{M}({\mathbf x})\phi_i.
$$
\end{proposition}

\begin{proof}
Indeed, by Proposition \ref{difftrick}, 
$$
	\frac{\partial}{\partial x_i}\mathrm{M}({\mathbf x})
	=\mathrm{M}({\mathbf x})\frac{\partial}{\partial y_i}\mathrm{K}({\mathbf y})
	=\mathrm{M}({\mathbf x})\phi_i.
$$
\end{proof}

\subsection{Effective action}\label{ssec:effective action}
We will assume from now on that the non-commutative random variables $a_1,\ldots,a_n$ used to define the coefficients of $\mathrm{M}({\mathbf x})$ are centered, that is, $\varphi(a_i)=0$,  $i=1,\ldots,n$. This implies in particular that $\mathrm{K}(\mathbf y)$ has no linear part and as a result that the conjugate field $\Phi$ has no constant part. Notice now that the linear part (in the variables $y_j$) of $\phi_i$ is, by definition (see \eqref{phi} in Definition~\ref{def:conjfield}), the sum $\mathrm{k}^{(2)}_{ij}y_j$. The field $\Phi$ defines therefore a (invertible, formal) change of coordinates with respect to $\mathbf y$ (and therefore also with respect to $\mathbf x$) if and only if the $n \times n$ matrix 
\begin{equation}
\label{covariance}
	{\mathcal K}:=(\mathrm{k}^{(2)}_{ij})_{1 \le i,j \le n}
\end{equation}
is invertible. As the matrix entries $\mathrm{k}^{(2)}_{ij}$ are the order-2 free cumulants, this matrix is the free probability analogue of the classical covariance matrix -- for a family $(V_1,\ldots,V_n)$ of real random variables. Assuming that ${\mathcal K}$ is invertible is therefore the free probability analogue of the assumption that the covariance matrix is invertible (which is the usual assumption made when studying for example the distribution of families of Gaussian variables). This justifies the naturalness of the following assumption:

\begin{definition}
\label{def:regular}
The conjugate field $\Phi=(\phi_1,\ldots,\phi_n)$, with components $\phi_i=\phi_i(\mathbf y)$, $1 \le i \le n$, is regular if and only if it has no constant part and the matrix  ${\mathcal K}$ is invertible. 
\end{definition}

As $\phi_i$ is the sum of $\mathrm{k}^{(2)}_{ij}y_j$ and higher oder non-commutative monomials (of degree $\geq 2$) in the $y_j$, the regularity assumption implies in particular that the $y_i$ can be expressed as formal power series in the $\phi_j$. 

\begin{definition} 
Given a regular conjugate field $\Phi$ of the integral field $\mathbf y$ associated to $\mathrm{M}$, the effective action (or Legendre transform of $\mathrm{K}$, see the proposition below) is the functional $\mathrm{L}\in R_\Phi$,
\begin{equation}
\label{Lengendre1}	
	\mathrm{L}(\Phi)=\mathrm{L}(\phi_1,\ldots,\phi_n)
\end{equation}
without constant terms, defined by one of the following equivalent formulas:
$$
	\frac{\partial}{\partial \phi_i}\mathrm{L}( \Phi)=y_i,
$$
$$
	\mathrm{L}(\Phi)=\phi_iy_i.
$$
\end{definition}

Explicitly, the effective action is a generating series in the components $\phi_i$ of the regular conjugate field $\Phi$
\begin{equation}
\label{Legendre2}	 
	\mathrm{L}(\Phi)= \sum\limits_{k>0  \atop 0< i_1,\ldots,i_k\leq n }
	\ell^{(k)}_{i_1 \cdots i_k}\phi_{i_1}\cdots \phi_{i_k},
\end{equation}
with so-called one-particle irreducible (1PI) coefficients of order $k$
\begin{equation}
\label{Legendre2a}	 
		\ell^{(k)}_{i_1 \cdots i_k}
		= \left. \frac{\partial^k}{\partial{\phi_{i_k}} \cdots \partial{\phi_{i_1}}}\right |_{\Phi=0}
	\mathrm{L}(\Phi).
\end{equation}

The following proposition shows that $\mathrm{L}$ is a non-commutative Legendre transform of $\mathrm{K}$, by analogy with the classical univariate case (for a strictly convex real function). Recall that, as $\mathrm{L}$ and $\mathrm{K}$ have no constant terms, they are entirely characterised by their corresponding differential fields.

\begin{proposition}
\label{prop:Legendre}
For a regular conjugate field $\Phi$ of the integral field $\mathbf y$, we have the non-commutative Legendre transform identities:
$$
	\partial \mathrm{L} \circ \partial \mathrm{K}({\mathbf y})
	=\mathbf y,
$$
and
$$
	\partial \mathrm{K} \circ \partial \mathrm{L}(\Phi)
	=\Phi,
$$
where the differential $\partial \mathrm{L}$ is taken with respect to $\Phi$ and the differential $\partial \mathrm{K}$ with respect to $\mathbf y$, i.e., $\partial \mathrm{L} =(\mathrm{L}(\Phi)_{\phi_1},\ldots,\mathrm{L}(\Phi)_{\phi_n})$ and $\partial \mathrm{K}=(\mathrm{K}({\mathbf y})_{y_1},\dots,\mathrm{K}({\mathbf y})_{y_n})$, respectively.
\end{proposition}

\begin{proof}
Indeed, we find
\begin{equation}
\label{Legendre3}	
	\frac{\partial}{\partial\phi_i}\mathrm{L}\Big(\frac{\partial}{\partial y_1} 
	\mathrm{K}({\mathbf y}),\ldots, \frac{\partial}{\partial y_n} \mathrm{K}({\mathbf y})\Big)
	=\frac{\partial}{\partial\phi_i}\mathrm{L}(\Phi)
	=y_i,
\end{equation}
and
\begin{equation}
\label{Legendre4}	
	\frac{\partial}{\partial y_i}\mathrm{K}\Big(\frac{\partial}{\partial \phi_1} 
	\mathrm{L}({\Phi}),\ldots, \frac{\partial}{\partial \phi_n} \mathrm{L}({\Phi})\Big)
	=\frac{\partial}{\partial y_i}\mathrm{K}({\mathbf y})
	=\phi_i.
\end{equation}
\end{proof}

Notice that the identities
\begin{equation}
\label{Legendre5}
	\mathrm{L}(\Phi)
	=\phi_iy_i,
	\qquad
	\mathrm{K}({\mathbf y})
	=y_i\phi_i,
\end{equation}
 clearly display the non-commutative nature of the non-commutative Legendre transform.
In the planar QFT literature, slightly different conventions appear in the definition of the Legendre transform of the connected planar Green's functions. The conjugate variables are also defined using partial derivatives of the planar connected Green's function, but sign conventions and the formulation of the identity defining the Legendre transform may vary. For example, when translated in our notation and conventions, the  fundamental non-commutative Legendre transform identity in Cvitanovich et al.~would read 
\begin{equation}
\label{Legendre5a}
	\mathrm{L}(\Phi) + \mathrm{K}({\mathbf y}) = \phi_iy_i + y_i\phi_i.
\end{equation} 
See \cite[Eq.~(40)]{Brezin}, \cite[Eqs.~(15) \& (16)]{cvitanovic1}, \cite[Eqs.~(4.16) \& (4.17)]{cvitanovic2}. The differential formulation we have chosen seems a natural choice, at least from the algebraic point of view we have adopted.

%%%%%%%%%%%%%%%%%%
%%%%%%%%%%%%%%%%%%

\section{Relations between effective action and free cumulants}
\label{ssec:realtions}

We will describe now in greater detail the relations between the 1PI coefficients $\ell^{(k)}_{i_1 \cdots i_k}$ (the ``$k$-point 1PI correlation functions'' in QFT parlance) in the effective action \eqref{Legendre2} and the free cumulants $\mathrm{k}^{(p)}_{i_1 \cdots i_p}$ (the ``$k$-point connected correlation functions'') in \eqref{eq:cumulants}.   We systematically assume from now on that the conjugate field $\Phi$ of $\mathbf y$ is regular.

%%%%%%%%%%%%%%%%%%

\subsection{2-point correlation functions identities}

For completeness, let us state some elementary but useful identities before turning to more advanced properties.

Using \eqref{phi}, the chain rule gives
$$
	\delta_{ji} =\frac{\partial }{\partial \phi_j} \phi_i 
	= \frac{\partial }{\partial \phi_j} (\frac{\partial}{\partial y_i} \mathrm{K}({\mathbf y}))
	= \frac{\partial y_l}{\partial \phi_j} \frac{\partial^2}{\partial y_l \partial y_i} \mathrm{K}({\mathbf y})
	=\frac{\partial y_l}{\partial \phi_j}  \mathrm{K}({\mathbf y})_{y_i y_l},
$$
while from \eqref{Legendre3} we obtain
$$
	 \frac{\partial y_l}{\partial \phi_j} 
	 = \frac{\partial^2}{\partial\phi_j\partial\phi_l}\mathrm{L}(\Phi) 
	 = \mathrm{L}(\Phi)_{\phi_l\phi_j}.
$$
Hence, we find
\begin{equation}
\label{matrix1}
	 \mathrm{L}(\Phi)_{\phi_l\phi_j}\mathrm{K}({\mathbf y})_{y_i y_l}= \delta_{ji},
\end{equation}
which may be understood as a (nontrivial) generating series refinement of the matrix identity $\mathcal K\mathcal K^{-1}=\text{Id}$. Indeed $ \mathrm{K}({\mathbf y})_{y_i y_l}$  has constant term
$$
	\mathrm{K}(0)_{y_i y_l} = {\mathcal K}_{il}=\mathrm{k}^{(2)}_{il}. 
$$
From
\begin{equation}
\label{matrix1bis}
	 \mathrm{K}(0)_{y_i y_l}\mathrm{L}(0)_{\phi_l\phi_j}= \delta_{ij},
\end{equation}
we get (recall that $\mathcal K$ is invertible since $\Phi$ is regular)
$$
	{\mathcal L}_{lj}:=\mathrm{L}(0)_{\phi_l\phi_j} = ({\mathcal K}^{-1})_{lj}.
$$

By symmetry in Legendre transform arguments, calculations dualise. As an illustration we show what happens when considering our previous identities.
$$
	\delta_{ji} =\frac{\partial }{\partial y_j} y_i 
	= \frac{\partial }{\partial y_j} (\frac{\partial}{\partial \phi_i} \mathrm{L}({\Phi}))
= \frac{\partial \phi_l}{\partial y_j} \frac{\partial^2}{\partial \phi_l \partial \phi_i} \mathrm{L}({\Phi})
	=\frac{\partial \phi_l}{\partial y_j}  \mathrm{L}({\Phi})_{\phi_i \phi_l},
$$
while 
$$
	 \frac{\partial \phi_l}{\partial y_j} 
	 = \frac{\partial^2}{\partial y_j\partial y_l}\mathrm{K}(\mathbf y) 
	 = \mathrm{K}(\mathbf y)_{y_ly_j},
$$
hence 
\begin{equation}
\label{matrix2}
	 \mathrm{K}({\mathbf y})_{y_l y_j}\mathrm{L}(\Phi)_{\phi_i\phi_l}= \delta_{ji}. 
\end{equation}

These relations between the two families of 2-point correlation functions (the $\mathrm{K}({\mathbf y})_{y_l y_j}$ and the $\mathrm{L}(\Phi)_{\phi_i\phi_l}$) lead to various other nontrivial identities typical of quantum field computations with Green's functions (of which various other examples can be found in the planar QFT literature, although the reader should be warned that they sometimes rely on extra hypothesis such as the existence of specific interaction terms such as $\phi^3$ of $\phi^4$).

Since the constant term of $\mathrm{K}({\mathbf y})_{y_l y_j}$, as a function of the $y_k$, is equal to its constant term, as a function of the $\phi_k$, \eqref{matrix2} and the planar Leibniz rule imply
$$
	0=\frac{\partial}{\partial \phi_k} \delta_{ji} 
	= \frac{\partial \mathrm{K}({\mathbf y})_{y_l y_j}}{\partial \phi_k} \mathrm{L}(\Phi)_{\phi_i\phi_l}
	+ \mathrm{K}(0)_{y_l y_j}\mathrm{L}(\Phi)_{\phi_i\phi_l\phi_k}.
$$
Multiplying from the right with $\mathrm{K}({\mathbf y})_{y_p y_i}$, which introduces an extra summation over the index $i$, we get
$$
	\frac{\partial \mathrm{K}({\mathbf y})_{y_l y_j}}{\partial \phi_k} 
	\mathrm{L}(\Phi)_{\phi_i\phi_l}\mathrm{K}({\mathbf y})_{y_p y_i}
	=- \mathrm{K}(0)_{y_l y_j}\mathrm{L}(\Phi)_{\phi_i\phi_l\phi_k}\mathrm{K}({\mathbf y})_{y_p y_i}.
$$
From \eqref{matrix1} we obtain
$$
	\frac{\partial \mathrm{K}({\mathbf y})_{y_l y_j}}{\partial \phi_k} 
	\mathrm{L}(\Phi)_{\phi_i\phi_l}\mathrm{K}({\mathbf y})_{y_p y_i}
	=\frac{\partial \mathrm{K}({\mathbf y})_{y_l y_j}}{\partial \phi_k}
	\delta_{pl}=\frac{\partial \mathrm{K}({\mathbf y})_{y_p y_j}}{\partial \phi_k},
$$
so that finally
\begin{equation}\label{3point}
	\frac{\partial \mathrm{K}({\mathbf y})_{y_p y_j}}{\partial \phi_k}
	=- \mathrm{K}(0)_{y_l y_j}\mathrm{L}(\Phi)_{\phi_i\phi_l\phi_k}\mathrm{K}({\mathbf y})_{y_p y_i}.
\end{equation}

%%%%%%%%%%%%%%%%%%

\subsection{Planar graphical calculus}
\label{ssec:planargraphcalc}

We will explore now systematically the relations between 1PI and connected correlation functions using a graphical calculus to handle them.

\label{ssec:pgc}
Applying the chain rule yields
\begin{equation}
\label{y2phi}
	\frac{\partial}{\partial y_i} 
	= \frac{\partial \phi_j}{\partial y_i} \frac{\partial }{\partial \phi_j} 
	= \frac{\partial^2}{\partial y_i\partial y_j} \mathrm{K}({\mathbf y})\frac{\partial }{\partial \phi_j} 
	= \mathrm{K}({\mathbf y})_{ y_j y_i}\frac{\partial }{\partial \phi_j}.
\end{equation}
Applying the planar Leibniz rule, the last equality yields
\begin{align}
\label{doublechain}
\begin{aligned}
	\frac{\partial}{\partial y_i} \frac{\partial}{\partial y_j} 
	&= \mathrm{K}({\mathbf y})_{ y_l y_i}\frac{\partial }{\partial \phi_l}
	   \mathrm{K}({\mathbf y})_{ y_m y_j}\frac{\partial }{\partial \phi_m}\\
	&= \mathrm{K}({\mathbf y})_{ y_l y_i}\frac{\partial \mathrm{K}({\mathbf y})_{ y_m y_j}}{\partial \phi_l}
	   \frac{\partial }{\partial \phi_m}
	   +\mathrm{K}({\mathbf y})_{ y_l y_i}
	   \mathrm{K}(0)_{ y_m y_j}\frac{\partial }{\partial \phi_l}\frac{\partial }{\partial \phi_m}\\   
	&\stackrel{\eqref{3point}}{=} -\mathrm{K}({\mathbf y})_{ y_l y_i}
		\mathrm{K}(0)_{y_k y_j}\mathrm{L}(\Phi)_{\phi_p\phi_k\phi_l}\mathrm{K}({\mathbf y})_{y_m y_p}
	\frac{\partial }{\partial \phi_m}   
	+\mathrm{K}({\mathbf y})_{ y_l y_i}
	   \mathrm{K}(0)_{ y_m y_j}\frac{\partial }{\partial \phi_l}\frac{\partial }{\partial \phi_m}.
\end{aligned}
\end{align}
When applying the last equality to the computation of 
$$ 
	\mathrm{K}({\mathbf y})_{y_{i_1}y_{i_2} y_{i_3}}
	=\frac{\partial^3}{\partial{y_{i_3}} \partial{y_{i_2}} \partial{y_{i_1}}}
	\mathrm{K}({\mathbf y})
	= \frac{\partial^{2}}{\partial{y_{i_3}}  \partial{y_{i_2}}}
	\phi_{i_1},
$$
we can ignore all higher order $\phi$-derivations because on the right-hand side the second order derivative, $\frac{\partial^{2}}{\partial{y_{i_3}}  \partial{y_{i_2}}}$, would be acting on $\phi_j$. In general, for the same reason, to compute $\mathrm{K}({\mathbf y})_{y_{i_1} \cdots y_{i_k}}$, it will be enough to consider only terms in the expansion of $\frac{\partial^{k-1}}{\partial{y_{i_k}} \cdots \partial{y_{i_2}}}$ that are linear in the differential operators $\frac{\partial }{\partial \phi_m}$. We will use freely this observation later on.

\medskip

We will now use a diagrammatical calculus to express $\mathrm{K}({\mathbf y})_{y_{i_1} \cdots y_{i_k}}$ as non-commutative polynomials in the $\mathrm{L}(\Phi)_{\phi_{j_1} \cdots \phi_{j_p}}$, $2 < p \le k$, the 2-point connected correlation functions $\mathrm{K}(0)_{ y_s y_t}$, and the 2-point connected correlation functionals $\mathrm{K}({\mathbf y})_{ y_i y_l}$. 

We start with 
\begin{equation}
\label{2pointfunction}
	\frac{\partial}{\partial y_2} \frac{\partial}{\partial y_1} \mathrm{K}({\mathbf y})
	= \mathrm{K}({\mathbf y})_{y_1y_2},
\end{equation}
which we represent graphically
\begin{equation}
\label{2point}
\begin{aligned}
\begin{tikzpicture}
\def \n {2}
\def \radius {1.5cm}
\foreach \s in {1,...,\n}
{
\draw[fill=blue]({360/\n * (\s - 1)}:\radius) circle (2 pt) node [above] {};
  \node at ({360/\n * (\s + 1)}:\radius + 0.2cm) {\tiny{$\s$}} ;
  \draw[-, >=latex] ({360/\n * (\s - 1)}:\radius) 
    arc ({360/\n * (\s - 1)}:{360/\n * (\s)}:\radius);
}

\foreach \from/\to in {1/2}
\draw ({360/\n * (\from - 1)}:\radius) -- ({360/\n * (\to - 1)}:\radius);

 %\draw[fill=black](0,0) circle (1 pt) node [above] {};
\end{tikzpicture}
\end{aligned}
\end{equation}
Next, we consider 
\begin{align*}
	\frac{\partial}{\partial y_3} \frac{\partial}{\partial y_2} \frac{\partial}{\partial y_1} \mathrm{K}({\mathbf y})
	&= \mathrm{K}({\mathbf y})_{y_1y_2y_3}\\
	&=\frac{\partial}{\partial y_3} \frac{\partial}{\partial y_2} \phi_1\\
	&= -\mathrm{K}({\mathbf y})_{ e_3 y_3}
		\mathrm{K}(0)_{e_2y_2}\mathrm{L}(\Phi)_{e_1e_2e_3}\mathrm{K}({\mathbf y})_{y_me_1}
	\frac{\partial }{\partial \phi_m}  \phi_1\\
	&= -\mathrm{K}({\mathbf y})_{ e_3 y_3}
		\mathrm{K}(0)_{e_2y_2}\mathrm{L}(\Phi)_{e_1e_2e_3}\mathrm{K}({\mathbf y})_{y_1e_1},
\end{align*}
which we represent graphically (up to the sign) 
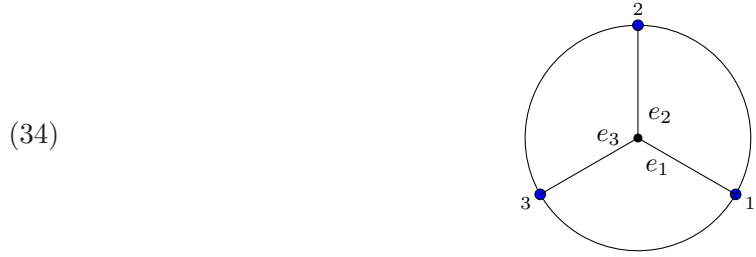
\begin{equation}
\label{3-graph}
\begin{aligned}
\begin{tikzpicture}
\def \n {3}
\def \radius {1.5cm}
\foreach \s in {1,...,\n}
{
\draw[fill=blue]({360/\n * (\s - 1) - 30}:\radius) circle (2 pt) node [above] {};
  \node at ({360/\n * (\s + 2) - 30}:\radius + 0.22cm) {\tiny{$\s$}} ;
  \draw[-, >=latex] ({360/\n * (\s - 1)}:\radius) 
    arc ({360/\n * (\s - 1)}:{360/\n * (\s)}:\radius);
}

 \draw[fill=black](0,0) circle (1.5 pt) node [above] {};

%\foreach \from/\to in {1/3}
%\draw ({360/\n * (\from - 1)}:\radius) -- ({360/\n * (\to - 1)}:\radius);

\foreach \from/\to in {1/3}
\draw({360/\n * (\from - 1) - 30}:\radius) -- node [pos=.8,below] {\small{$e_1$}}  (0,0);

\foreach \from/\to in {2/3}
\draw ({360/\n * (\from - 1) - 30}:\radius) -- node [pos=.8,right] {\small{$e_2$}} (0,0) ;

\foreach \from/\to in {3/3}
\draw ({360/\n * (\from - 1) - 30}:\radius) -- node [pos=.7,above] {\small{$e_3$}} (0,0);
\end{tikzpicture}
\end{aligned}
\end{equation}

In this representation, starting from the rightmost term of the expression 
$$
	\mathrm{K}({\mathbf y})_{ e_3 y_3}
		\mathrm{K}(0)_{e_2y_2}\mathrm{L}(\Phi)_{e_1e_2e_3}\mathrm{K}({\mathbf y})_{y_1e_1},
$$
the $\mathrm{K}({\mathbf y})_{ab}$ are represented as edges with endpoints decorated by $a$ and $b$ and $\mathrm{L}(\Phi)_{e_1e_2e_3}$ by a vertex with three outcoming half-edges decorated counterclockwise by $e_1,e_2,e_3$. The ``Feynman rules'' that associate graphs to functionals, and conversely, will be explained in detail later on. 

\begin{remark}
Note that we simplified notation by introducing new dummy variable symbols $e_j$. They correspond to Einstein summations and express joint derivations with respect to $y_j$ and $\phi_j$ on $\mathrm{K}({\mathbf y})$ respectively $\mathrm{L}(\Phi)$. For example 
$$
	\mathrm{K}(0)_{y_j e_2}\mathrm{L}(\Phi)_{e_2\phi_q}
	:=\sum\limits_{k=1}^n\mathrm{K}(0)_{y_j y_k}\mathrm{L}(\Phi)_{\phi_k\phi_q}.
$$
These symbols will prove especially useful in the graphical representation of 1PI vertex terms such as the 3-vertex functional
$$
	\mathrm{K}({\mathbf y})_{ e_3 y_3}
		\mathrm{K}(0)_{e_2y_2}\mathrm{L}(\Phi)_{e_1e_2e_3}\mathrm{K}({\mathbf y})_{y_1e_1},
$$
whose notation we will below further simplify by writing
$$
	\mathrm{K}({\mathbf y})_{ e_3 3}
		\mathrm{K}(0)_{e_22}\mathrm{L}(\Phi)_{e_1e_2e_3}\mathrm{K}({\mathbf y})_{1e_1}.
$$
\end{remark}
Besides, we will also use the notation $\cong$ to denote that two differential operators are equal up to elements (linear combinations of iterated products of $\frac{\partial}{\partial \phi_j}$)
that vanish when applied to the $\phi_i$'s.

Next we compute the order four term, $\mathrm{K}({\mathbf y})_{y_1y_2y_3y_4}$, using lower order computations
$$
	 \mathrm{K}({\mathbf y})_{y_1y_2y_3y_4}
	 =\frac{\partial}{\partial y_4}\frac{\partial}{\partial y_3} \frac{\partial}{\partial y_2} 
	 \frac{\partial}{\partial y_1} \mathrm{K}({\mathbf y})
	 =\frac{\partial}{\partial y_4}\frac{\partial}{\partial y_3} \frac{\partial}{\partial y_2} 
	\phi_1.
$$
Using the  planar chain and Leibniz rules we get  
\begin{align*}
	\lefteqn{
	\frac{\partial}{\partial y_4}\frac{\partial}{\partial y_3} \frac{\partial}{\partial y_2} 
	\cong \frac{\partial}{\partial y_4}\Big(
	-\mathrm{K}({\mathbf y})_{ e_3 3}
		\mathrm{K}(0)_{e_22}\mathrm{L}(\Phi)_{e_1e_2e_3}\mathrm{K}({\mathbf y})_{y_me_1}
		\frac{\partial }{\partial \phi_m} \Big)  }\\
&=\mathrm{K}({\mathbf y})_{y_j4}\frac{\partial }{\partial \phi_j}\Big(
	-\mathrm{K}({\mathbf y})_{ e_3 3}
		\mathrm{K}(0)_{e_22}\mathrm{L}(\Phi)_{e_1e_2e_3}\mathrm{K}({\mathbf y})_{y_me_1}
		\frac{\partial }{\partial \phi_m} \Big)\\
&\cong\mathrm{K}({\mathbf y})_{e_44}\Big(
		\mathrm{K}(0)_{e_53}\mathrm{L}(\Phi)_{e_6e_5e_4}\mathrm{K}({\mathbf y})_{e_3e_6}
		\mathrm{K}(0)_{e_22}\mathrm{L}(\Phi)_{e_1e_2e_3}\mathrm{K}({\mathbf y})_{y_me_1}
		\frac{\partial }{\partial \phi_m}\\
&-    \mathrm{K}(0)_{ e_3 3}
		\mathrm{K}(0)_{e_22}\mathrm{L}(\Phi)_{e_1e_2e_3e_4}\mathrm{K}({\mathbf y})_{y_me_1}
		\frac{\partial }{\partial \phi_m}\\
	&+\mathrm{K}(0)_{ e_3 3}
		\mathrm{K}(0)_{e_22}\mathrm{L}(\Phi)_{e_1e_2e_3}
		\mathrm{K}(0)_{e_5e_1}\mathrm{L}(\Phi)_{e_6e_5e_4}\mathrm{K}({\mathbf y})_{y_me_6}
		\frac{\partial }{\partial \phi_m}\Big).
\end{align*}
This leads to a decomposition of $\mathrm{K}({\mathbf y})_{y_1y_2y_3y_4}$ as a sum of three terms that
 we can graphically represent as follows. The first term
\begin{align*}
	\lefteqn{
	\mathrm{K}({\mathbf y})_{e_44}
		\mathrm{K}(0)_{e_53}\mathrm{L}(\Phi)_{e_6e_5e_4}\mathrm{K}({\mathbf y})_{e_3e_6}
		\mathrm{K}(0)_{e_22}\mathrm{L}(\Phi)_{e_1e_2e_3}\mathrm{K}({\mathbf y})_{y_me_1}
		\frac{\partial }{\partial \phi_m}\phi_1}\\
	&=\mathrm{K}({\mathbf y})_{e_44}
		\mathrm{K}(0)_{e_53}\mathrm{L}(\Phi)_{e_6e_5e_4}\mathrm{K}({\mathbf y})_{e_3e_6}
		\mathrm{K}(0)_{e_22}\mathrm{L}(\Phi)_{e_1e_2e_3}\mathrm{K}({\mathbf y})_{1e_1}
\end{align*}
corresponds to 
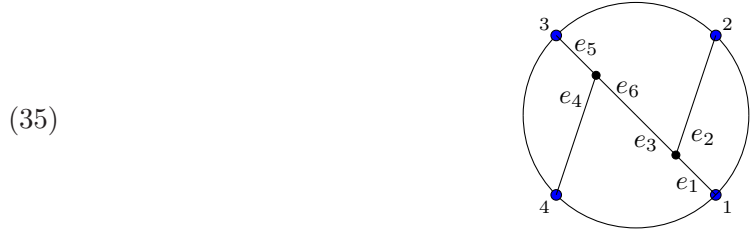
\begin{equation}
\label{4point2}
\begin{aligned}
\begin{tikzpicture}
\def \n {4}
\def \radius {1.5cm}
\foreach \s in {1,...,\n}
{
\draw[fill=blue]({360/\n * (\s - 1)- 45}:\radius) circle (2 pt) node [above] {};
  \node at ({360/\n * (\s + 3)- 45}:\radius + 0.22cm) {\tiny{$\s$}} ;
  \draw[-, >=latex] ({360/\n * (\s - 1)}:\radius) 
    arc ({360/\n * (\s - 1)}:{360/\n * (\s)}:\radius);
}

%\draw[fill=black](0,0) circle (1.5 pt) node [above] {};

%\foreach \from/\to in {1/3}
%\draw ({360/\n * (\from - 1)}:\radius) -- ({360/\n * (\to - 1)}:\radius);

\foreach \from/\to in {1/3}
\draw ({360/\n * (\from - 1)- 45}:\radius) -- node [pos=.7,below] {\small{$e_1$}} ({360/\n * (\from - 1)- 45}:\radius/2);

\draw ({360/\n * 0 - 45}:\radius/2) -- node [pos=.75,below] {\small{$e_3$}} (0,0);

 \draw[fill=black]({360/\n * 0 - 45}:\radius/2) circle (1.5 pt) node [above] {};

\foreach \from/\to in {2/3}
\draw ({360/\n * (\from - 1)- 45}:\radius) -- node [pos=.7,below] {\hspace{0.4cm}\small{$e_2$}} ({360/\n * 0 - 45}:\radius/2) ;

\foreach \from/\to in {3/3}
\draw ({360/\n * (\from - 1)- 45}:\radius) -- node [pos=.75,above] {\small{$e_5$}} ({360/\n * (\from - 1)- 45}:\radius/2) ;

\draw ({360/\n * 2 - 45}:\radius/2) -- node [pos=.8,above] {\small{$e_6$}} (0,0);

\draw[fill=black]({360/\n * 2 - 45}:\radius/2) circle (1.5 pt) node [above] {};

\foreach \from/\to in {4/3}
\draw ({360/\n * (\from - 1)- 45}:\radius) -- node [pos=.65,above] {\hspace{-0.3cm}\small{$e_4$}} ({360/\n * 2 - 45}:\radius/2);
\end{tikzpicture}
\end{aligned}
\end{equation}
which is obtained from the diagram \eqref{3-graph} by the local graphical operation at the leaf vertex labeled $3$\\
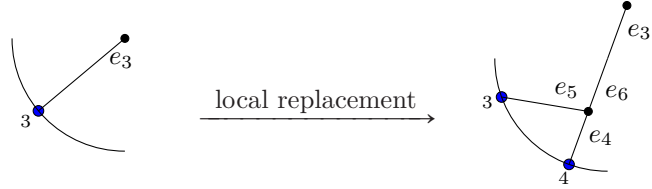
$$
\begin{array}{ccc}
\begin{tikzpicture}
\def \n {1}
\def \radius {1.5cm}
\foreach \s in {1,...,\n}
{
\draw[fill=blue]({360/\n * (\s - 1) - 140}:\radius) circle (2 pt) node [above] {};
  \node at ({360/\n * (\s + 2) - 140}:\radius + 0.22cm) {\tiny{$3$}} ;
  \draw[-, >=latex] ({360/\n * (\s)}:-\radius) 
    arc ({180/\n * (\s)}:{270/\n * (\s)}:\radius);
}

 \draw[fill=black](0,0) circle (1.5 pt) node [above] {};

\foreach \from/\to in {1/3}
\draw({360/\n * (\from - 1) - 140}:\radius) -- node [pos=.7,right] {\small{$e_3$}}  (0,0);

\end{tikzpicture}
\quad & & \\[-1cm]
 & \xrightarrow{\makebox{ {\small{local replacement}} }}\ & \\[-2cm]
& &
\begin{tikzpicture}
\def \radius {1.5cm}

  \draw[fill=blue](200:\radius) circle (2 pt) node [above] {};
  \draw[fill=blue](250:\radius) circle (2 pt) node [above] {};
  
  \node at (200:\radius + 0.2cm) {\tiny{$3$}} ;
  \node at (250:\radius + 0.2cm) {\tiny{$4$}} ;
 
  \draw[-, >=latex] ({180}:\radius) arc ({180}:{270}:\radius);

\draw[fill=black](250:-\radius/2) circle (1.5 pt) node [above] {};
\draw[fill=black](250:\radius/2) circle (1.5 pt) node [above] {};

\draw ({200}:\radius) -- node [pos=.75,above] {\small{$e_5$}} (250:\radius/2);
\draw ({250}:\radius) -- node [pos=.5,right] {\small{$e_4$}}  (250:\radius/2);
\draw (0,0) -- node [pos=.65,right] {\small{$e_6$}}  (250:\radius/2);
\draw (0,0) -- node [pos=.5,right] {\small{$e_3$}}  (70:\radius/2);

\end{tikzpicture}
\end{array}
$$
This local replacement encodes, with our previous notation, the action of $\frac{\partial}{\partial \phi_j}$ on $\mathrm{K}({\mathbf y})_{e_33}$. Similar observations hold for the following terms.

The next term
$$
	-\mathrm{K}(0)_{ e_3 3}
		\mathrm{K}(0)_{e_22}\mathrm{L}(\Phi)_{e_1e_2e_3e_4}
		\mathrm{K}({\mathbf y})_{1e_1}$$
has the graphical representation
$$
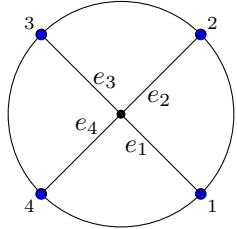
\begin{tikzpicture}
\def \n {4}
\def \radius {1.5cm}
\foreach \s in {1,...,\n}
{
\draw[fill=blue]({360/\n * (\s - 1)- 45}:\radius) circle (2 pt) node [above] {};
  \node at ({360/\n * (\s + 3)- 45}:\radius + 0.22cm) {\tiny{$\s$}} ;
  \draw[-, >=latex] ({360/\n * (\s - 1)}:\radius) 
    arc ({360/\n * (\s - 1)}:{360/\n * (\s)}:\radius);
}

 \draw[fill=black](0,0) circle (1.5 pt) node [above] {};

%\foreach \from/\to in {1/3}
%\draw ({360/\n * (\from - 1)}:\radius) -- ({360/\n * (\to - 1)}:\radius);

\foreach \from/\to in {1/3}
\draw({360/\n * (\from - 1) - 45}:\radius) -- node [pos=.8,below] {\small{$e_1$}}  (0,0);

\foreach \from/\to in {2/3}
\draw ({360/\n * (\from - 1) - 45}:\radius) -- node [pos=.8,right] {\small{$e_2$}} (0,0) ;

\foreach \from/\to in {3/3}
\draw ({360/\n * (\from - 1) - 45}:\radius) -- node [pos=.8,above] {\small{$e_3$}} (0,0);

\foreach \from/\to in {4/3}
\draw ({360/\n * (\from - 1) - 45}:\radius) -- node [pos=.85,left] {\small{$e_4$}} (0,0);

\end{tikzpicture}
$$
which is obtained from \eqref{3-graph} by adding a new external vertex leaf labeled $4$ and connecting it to the internal 3-vertex $\mathrm{L}(\Phi)_{e_1e_2e_3}$\\
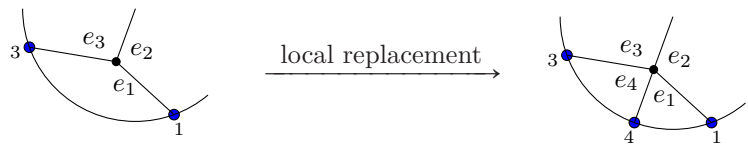
$$
\begin{array}{ccc}
\begin{tikzpicture}
\def \radius {1.5cm}

\draw[fill=blue](200:\radius) circle (2 pt) node [above] {}; 
\draw[fill=blue](290:\radius) circle (2 pt) node [above] {};
   
  \node at (200:\radius + 0.2cm) {\tiny{$3$}} ;
  \node at (290:\radius + 0.2cm) {\tiny{$1$}} ;
 
  \draw[-, >=latex] ({180}:\radius) arc ({180}:{310}:\radius);

\draw[fill=black](250:\radius/2) circle (1.5 pt) node [above] {};

\draw ({200}:\radius) -- node [pos=.75,above] {\small{$e_3$}} (250:\radius/2);
\draw ({290}:\radius) -- node [pos=.85,below] {\small{$e_1$}}  (250:\radius/2);
\draw (0,0) -- node [pos=.8,right] {\small{$e_2$}}  (250:\radius/2);

\end{tikzpicture}
\quad & & \\[-1.5cm]
 & \xrightarrow{\makebox{ {\small{local replacement}} }}\ & \\[-1cm]
& &
\begin{tikzpicture}
\def \radius {1.5cm}

\draw[fill=blue](200:\radius) circle (2 pt) node [above] {}; 
\draw[fill=blue](250:\radius) circle (2 pt) node [above] {};
\draw[fill=blue](290:\radius) circle (2 pt) node [above] {};
   
  \node at (200:\radius + 0.2cm) {\tiny{$3$}} ;
  \node at (250:\radius + 0.2cm) {\tiny{$4$}} ;
  \node at (290:\radius + 0.2cm) {\tiny{$1$}} ;
 
  \draw[-, >=latex] ({180}:\radius) arc ({180}:{310}:\radius);

\draw[fill=black](250:\radius/2) circle (1.5 pt) node [above] {};

\draw ({200}:\radius) -- node [pos=.75,above] {\small{$e_3$}} (250:\radius/2);
\draw ({250}:\radius) -- node [pos=.75,left] {\small{$e_4$}} (250:\radius/2);
\draw ({290}:\radius) -- node [pos=.8,below] {\small{$e_1$}}  (250:\radius/2);
\draw (0,0) -- node [pos=.8,right] {\small{$e_2$}}  (250:\radius/2);

\end{tikzpicture}
\end{array}
$$
In this process, the 3-vertex, representing $\mathrm{L}(\Phi)_{e_1e_2e_3}$, is changed into a 4-vertex, standing for $\mathrm{L}(\Phi)_{e_1e_2e_3e_4}$
$$
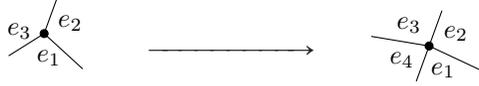
\begin{array}{ccc}
\begin{tikzpicture}
\def \radius {1.cm}

\draw[fill=black](250:\radius/2) circle (1.5 pt) node [above] {};

\draw ({230}:\radius) -- node [pos=.3,above] {\small{$e_3$}} (250:\radius/2);
\draw ({290}:\radius) -- node [pos=.3,left] {\small{$e_1$}}  (250:\radius/2);
\draw (0,0) -- node [pos=.7,right] {\small{$e_2$}}  (250:\radius/2);

\end{tikzpicture}
\quad & & \\[-0.7cm]
 & \xrightarrow{\hspace{2cm}}\quad & \\[-0.8cm]
& &
\begin{tikzpicture}
\def \radius {1.cm}

\draw[fill=black](250:\radius/2) circle (1.5 pt) node [above] {};

\draw ({200}:\radius) -- node [pos=.65,above] {\small{$e_3$}} (250:\radius/2);
\draw ({250}:\radius) -- node [pos=.6,left] {\small{$e_4$}}  (250:\radius/2);
\draw ({305}:\radius) -- node [pos=.75,below] {\small{$e_1$}}  (250:\radius/2);
\draw (0,0) -- node [pos=.65,right] {\small{$e_2$}}  (250:\radius/2);

\end{tikzpicture}
\end{array}
$$
The diagrammatical representation of the third term 
$$\mathrm{K}(0)_{ e_3 3}
		\mathrm{K}(0)_{e_22}\mathrm{L}(\Phi)_{e_1e_2e_3}
		\mathrm{K}(0)_{e_5e_1}\mathrm{L}(\Phi)_{e_6e_5e_4}\mathrm{K}({\mathbf y})_{1e_6}$$
is similar to \eqref{4point2}
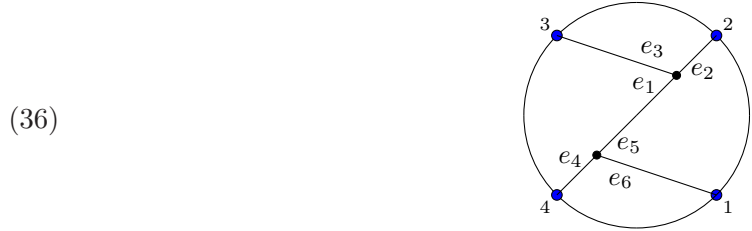
\begin{equation}
\label{4point3}
\begin{aligned}
\begin{tikzpicture}
\def \n {4}
\def \radius {1.5cm}
\foreach \s in {1,...,\n}
{
\draw[fill=blue]({360/\n * (\s - 1)- 45}:\radius) circle (2 pt) node [above] {};
  \node at ({360/\n * (\s + 3)- 45}:\radius + 0.22cm) {\tiny{$\s$}} ;
  \draw[-, >=latex] ({360/\n * (\s - 1)}:\radius) 
    arc ({360/\n * (\s - 1)}:{360/\n * (\s)}:\radius);
}

\foreach \from/\to in {1/3}
\draw ({360/\n * (\from - 1)- 45}:\radius) -- node [pos=.8,below] {\small{$e_6$}} ({360/\n * 3 - 45}:\radius/2);

\draw ({360/\n * 1 - 45}:\radius/2) -- node [pos=1.75,right] {\small{$e_5$}} (0,0);

 \draw[fill=black]({360/\n * 1 - 45}:\radius/2) circle (1.5 pt) node [above] {};

\foreach \from/\to in {2/3}
\draw ({360/\n * (\from - 1)- 45}:\radius) -- node [pos=.9,right] {\small{$e_2$}} ({360/\n * 1 - 45}:\radius/2) ;
 
\foreach \from/\to in {3/3}
\draw ({360/\n * (\from - 1)- 45}:\radius) -- node [pos=.8,above] {\small{$e_3$}} ({360/\n * 1 - 45}:\radius/2) ;

\draw ({360/\n * 3 - 45}:\radius/2) -- node [pos=1.75,left] {\small{$e_1$}} (0,0);

\foreach \from/\to in {4/3}
\draw ({360/\n * (\from - 1)- 45}:\radius) -- node [pos=.9,left] {\hspace{-0.3cm}\small{$e_4$}} ({360/\n * 3 - 45}:\radius/2);

\draw[fill=black]({360/\n * 3 - 45}:\radius/2) circle (1.5 pt) node [above] {};

\end{tikzpicture}
\end{aligned}
\end{equation}
It is obtained from \eqref{3-graph} by adding locally an external vertex leaf labeled $4$ together with an internal 3-vertex $\mathrm{L}(\Phi)_{e_6e_5e_4}$\\
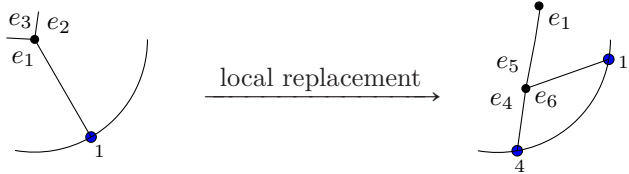
$$
\begin{array}{ccc}
\begin{tikzpicture}
\def \radius {1.5cm}
{
\draw[fill=blue]({0 - 60}:\radius) circle (2 pt) node [above] {};
  \node at ({0 - 60}:\radius + 0.22cm) {\tiny{$1$}} ;
  \draw[-, >=latex] (1.5,0) arc ({360}:{260}:\radius);
}

 \draw[fill=black](0,0) circle (1.5 pt) node [above] {};

\draw({0 - 60}:\radius) -- node [pos=.8,left] {\small{$e_1$}}  (0,0);
\draw({0 - 98}:-\radius/4) -- node [right] {\small{$e_2$}}  (0,0);
\draw({0 - 3}:-\radius/4) -- node [above] {\small{$e_3$}}  (0,0);

\end{tikzpicture}
\quad & & \\[-1.5cm]
 & \xrightarrow{\makebox{ {\small{local replacement}} }}\ & \\[-1.7cm]
& &
\begin{tikzpicture}
\def \radius {1.5cm}

  \draw[fill=blue](280:\radius) circle (2 pt) node [above] {};
  \draw[fill=blue](350:\radius) circle (2 pt) node [above] {};
  
  \node at (280:\radius + 0.2cm) {\tiny{$4$}} ;
  \node at (350:\radius + 0.2cm) {\tiny{$1$}} ;
 
  \draw[-, >=latex] (1.5,0) arc ({360}:{260}:\radius);

\draw[fill=black](0.55,0.45) circle (1.5 pt) node [above] {};
\draw[fill=black](300:\radius/2) circle (1.5 pt) node [above] {};

\draw ({280}:\radius) -- node [pos=.8,left] {\small{$e_4$}} (300:\radius/2);
\draw ({350}:\radius) -- node [pos=.75,below] {\small{$e_6$}}  (300:\radius/2);
\draw (0.5,0) -- node [pos=.5,left] {\small{$e_5$}}  (300:\radius/2);
\draw (0.5,0) -- node [pos=.5,right] {\small{$e_1$}}  (40:\radius/2);

\end{tikzpicture}
\end{array}
$$

We can generalise these example computations to a set of general rules. 

Notice first that indices of the dummy variables can be selected arbitrarily in our previous computations and play no role. For example, the following graph has exactly the same meaning as \ref{4point3}:
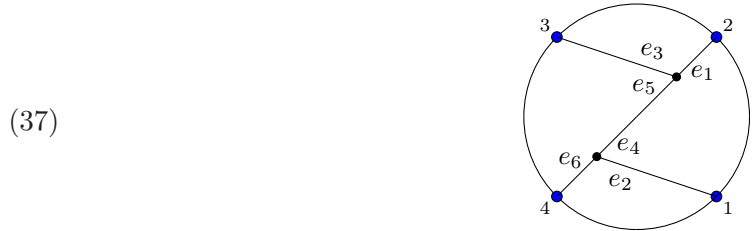
\begin{equation}
\label{4point32}
\begin{aligned}
\begin{tikzpicture}
\def \n {4}
\def \radius {1.5cm}
\foreach \s in {1,...,\n}
{
\draw[fill=blue]({360/\n * (\s - 1)- 45}:\radius) circle (2 pt) node [above] {};
  \node at ({360/\n * (\s + 3)- 45}:\radius + 0.22cm) {\tiny{$\s$}} ;
  \draw[-, >=latex] ({360/\n * (\s - 1)}:\radius) 
    arc ({360/\n * (\s - 1)}:{360/\n * (\s)}:\radius);
}

\foreach \from/\to in {1/3}
\draw ({360/\n * (\from - 1)- 45}:\radius) -- node [pos=.8,below] {\small{$e_2$}} ({360/\n * 3 - 45}:\radius/2);

\draw ({360/\n * 1 - 45}:\radius/2) -- node [pos=1.75,right] {\small{$e_4$}} (0,0);

 \draw[fill=black]({360/\n * 1 - 45}:\radius/2) circle (1.5 pt) node [above] {};

\foreach \from/\to in {2/3}
\draw ({360/\n * (\from - 1)- 45}:\radius) -- node [pos=.9,right] {\small{$e_1$}} ({360/\n * 1 - 45}:\radius/2) ;
 
\foreach \from/\to in {3/3}
\draw ({360/\n * (\from - 1)- 45}:\radius) -- node [pos=.8,above] {\small{$e_3$}} ({360/\n * 1 - 45}:\radius/2) ;

\draw ({360/\n * 3 - 45}:\radius/2) -- node [pos=1.75,left] {\small{$e_5$}} (0,0);

\foreach \from/\to in {4/3}
\draw ({360/\n * (\from - 1)- 45}:\radius) -- node [pos=.9,left] {\hspace{-0.3cm}\small{$e_6$}} ({360/\n * 3 - 45}:\radius/2);

\draw[fill=black]({360/\n * 3 - 45}:\radius/2) circle (1.5 pt) node [above] {};

\end{tikzpicture}
\end{aligned}
\end{equation}
It follows that these variables can be safely omitted from graphical representations without altering the meaning of the graphical encoding of functionals. 

\begin{definition}
\label{def:graphicalrules}
A generic admissible tree $T$ is defined to be a planar rooted tree 
 \begin{itemize}
 \item with a root decorated by $1$ and $n-1$ leaves decorated from $2$ to $n$,
 \item which can be inscribed in a disk, with the leaves located on the boundary of the disk and the decorations in the natural anti-clockwise order,
 \item with internal vertices in the interior of the disk and of arity at least $3$.
 \end{itemize}
Given a word $w=i_1\cdots i_n$ with letters in $[n]$, a $w$-admissible tree is a generic admissible tree in which the decorations $1,\dots ,n$ are replaced by $i_1,\dots,i_n$. The set of generic admissible trees is denoted $\mathcal{D}$, the set of $w$-admissible trees $\mathcal{D}_w$
\end{definition}

\begin{remark}Recall that a Schr\"oder tree is a planar rooted tree whose internal vertices have at least two descendants. As a Schr\"oder tree can be inscribed in a disk and its leaves labelled in the natural anti-clockwise order, the set of generic admissible trees is in bijection with the set $ST$ of Schr\"oder trees and the set of $w$-admissible trees in bijection with the set $ST(n)$ of Schr\"oder trees with $n-1$ leaves.
\end{remark}

Notice that if one erases in a generic admissible tree with $n-1$ leaves the leaf labelled $n$  and its outcoming edge (and the terminal vertex of this edge if and only if this vertex is of arity 3), one is left with a generic admissible tree with $n-2$ leaves.

The key observation of our following developments is that the action of planar derivatives $\frac{\partial}{\partial \phi_i}$ translates into a recursive process of generation of admissible trees, so that there is a perfect match between
\begin{itemize}
\item the recursive computation of the $n$ point connected correlation functionals in terms of 2-point correlation functions and functionals, and
\item the generation of admissible trees with $n-1$ leaves from admissible trees with $n-2$ leaves.
\end{itemize}

We first describe the recursive generation of admissible graphs.
For notational simplicity we consider hereafter only generic admissible trees --- the same reasoning would apply in general.
\begin{enumerate}
\item 
If $n=2$, then \eqref{2point} is the only tree.

\item
If $n>2$, the tree $T$ can be obtained from an admissible tree $T'$ with $n-2$ leaves as follows: consider the lowest path joining leaf $1$ to leaf $n-1$. This path is called the characteristic path.

$$
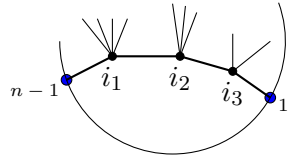
\begin{tikzpicture}
\def \radius {1.5cm}

\draw[fill=blue](200:\radius) circle (2 pt) node [above] {}; 
\draw[fill=blue](330:\radius) circle (2 pt) node [above] {};
   
  \node at (200:\radius + 0.45cm) {\tiny{$n-1$}} ;
  \node at (330:\radius + 0.2cm) {\tiny{$1$}} ;
 
  \draw[-, >=latex] ({180}:\radius) arc ({180}:{380}:\radius);

\draw[fill=black](0.1,-0.2) circle (1.5 pt) node [below] {$i_2$};
\draw[fill=black](-0.8,-0.2) circle (1.5 pt) node [below] {$i_1$};
\draw[fill=black](0.8,-0.4) circle (1.5 pt) node [below] {$i_3$};

\draw[thick] ({200}:\radius) -- (-0.8,-0.2);
\draw[thick] ({330}:\radius) -- (0.8,-0.4);
\draw[thick] (0.1,-0.2) -- (-0.8,-0.2);
\draw[thick] (0.1,-0.2) -- (0.8,-0.4);

\draw (0.1,-0.2) -- (90:\radius/3);
\draw (0.1,-0.2) -- (70:\radius/3);
\draw (0.1,-0.2) -- (50:\radius/3);

\draw (-0.8,-0.2) -- (-0.6,0.3);
\draw (-0.8,-0.2) -- (-1,0.3);
\draw (-0.8,-0.2) -- (-1.2,0.3);
\draw (-0.8,-0.2) -- (-0.8,0.3);

\draw (0.8,-0.4) -- (0.8,0.1);
\draw (0.8,-0.4) -- (1.3,0.);
%\draw (0.8,-0.4) -- (1.2,0.1);

\end{tikzpicture}
$$

The tree $T$ is then obtained from $T'$ either by

\begin{itemize}

\item adding a leaf decorated with $n$ and joining it to one of the internal vertices $i_1$, $i_2$, ... $i_p$. For instance
$$
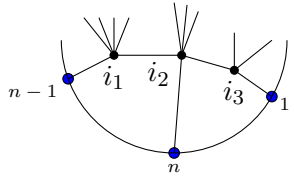
\begin{tikzpicture}
\def \radius {1.5cm}

\draw[fill=blue](200:\radius) circle (2 pt) node [above] {}; 
\draw[fill=blue](270:\radius) circle (2 pt) node [above] {};
\draw[fill=blue](330:\radius) circle (2 pt) node [above] {};
   
  \node at (200:\radius + 0.5cm) {\tiny{$n-1$}} ;
  \node at (270:\radius + 0.2cm) {\tiny{$n$}} ;
  \node at (330:\radius + 0.2cm) {\tiny{$1$}} ;
 
  \draw[-, >=latex] ({180}:\radius) arc ({180}:{360}:\radius);

\draw[fill=black](0.1,-0.2) circle (1.5 pt);
\draw[fill=black](-0.8,-0.2) circle (1.5 pt) node [below] {$i_1$};
\draw[fill=black](0.8,-0.4) circle (1.5 pt) node [below] {$i_3$};

\draw ({200}:\radius) -- (-0.8,-0.2);
\draw ({270}:\radius) -- node [pos=.81,left] {$i_2$} (0.1,-0.2);
\draw ({330}:\radius) -- (0.8,-0.4);
\draw (0.1,-0.2) -- (-0.8,-0.2);
\draw (0.1,-0.2) -- (0.8,-0.4);

\draw (0.1,-0.2) -- (90:\radius/3);
\draw (0.1,-0.2) -- (70:\radius/3);
\draw (0.1,-0.2) -- (50:\radius/3);

\draw (-0.8,-0.2) -- (-0.6,0.3);
\draw (-0.8,-0.2) -- (-1,0.3);
\draw (-0.8,-0.2) -- (-1.2,0.3);
\draw (-0.8,-0.2) -- (-0.8,0.3);
%\draw (-0.8,-0.2) -- (-0.9,0.4);

\draw (0.8,-0.4) -- (0.8,0.1);
\draw (0.8,-0.4) -- (1.3,0.);
%\draw (0.8,-0.4) -- (1.2,0.1);

\end{tikzpicture}
$$
In that case, the new characteristic path would be
$$
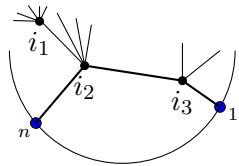
\begin{tikzpicture}
\def \radius {1.5cm}

\draw[fill=blue](220:\radius) circle (2 pt) node [above] {};
\draw[fill=blue](330:\radius) circle (2 pt) node [above] {};
   
  \node at (220:\radius + 0.2cm) {\tiny{$n$}} ;
  \node at (330:\radius + 0.2cm) {\tiny{$1$}} ;
 
  \draw[-, >=latex] ({180}:\radius) arc ({180}:{360}:\radius);

\draw[fill=black](-0.5,-0.2) circle (1.5 pt) node [below] {$i_2$};
\draw[fill=black](0.8,-0.4) circle (1.5 pt) node [below] {$i_3$};
\draw[fill=black](-1.1,0.39) circle (1.5 pt) node [below] {$i_1$};

\draw[thick] ({220}:\radius) -- (-0.5,-0.2);
\draw[thick] ({330}:\radius) -- (0.8,-0.4);
\draw[thick] (-0.5,-0.2) -- (0.8,-0.4);

\draw (-0.5,-0.2) -- (160:\radius/1.3);
\draw (-0.5,-0.2) -- (150:\radius/1.5);
\draw (-0.5,-0.2) -- (145:\radius/2);
\draw (-0.5,-0.2) -- (140:\radius/2.8);

\draw (0.8,-0.4) -- (0.8,0.1);
\draw (0.8,-0.4) -- (1.3,0.);
%\draw (0.8,-0.4) -- (1.2,0.1);

\draw (-1.1,0.39) -- (-1.1,0.59);
\draw (-1.1,0.39) -- (-1.,0.59);
\draw (-1.1,0.39) -- (-1.25,0.6);
\draw (-1.1,0.39) -- (-1.4,0.5);
\draw (-1.1,0.39) -- (-1.45,0.4);

\end{tikzpicture}
$$

\item selecting an edge of the characteristic path and making the following substitution
$$
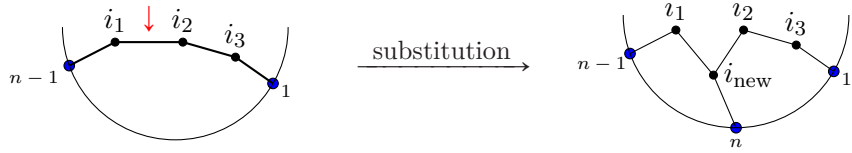
\begin{array}{ccc}
\begin{tikzpicture}
\def \radius {1.5cm}

\draw[fill=blue](200:\radius) circle (2 pt) node [above] {}; 
\draw[fill=blue](330:\radius) circle (2 pt) node [above] {};
   
  \node at (200:\radius + 0.5cm) {\tiny{$n-1$}} ;
  \node at (330:\radius + 0.2cm) {\tiny{$1$}} ;
 
  \draw[-, >=latex] ({180}:\radius) arc ({180}:{360}:\radius);

\draw[fill=black](0.1,-0.2) circle (1.5 pt) node [above] {$i_2$};
\draw[fill=black](-0.8,-0.2) circle (1.5 pt) node [above] {$i_1$};
\draw[fill=black](0.8,-0.4) circle (1.5 pt) node [above] {$i_3$};

\draw[thick] ({200}:\radius) -- (-0.8,-0.2);
\draw[thick] ({330}:\radius) -- (0.8,-0.4);
\draw[thick] (0.1,-0.2) -- node [above] {$\textcolor{red}{\downarrow}$} (-0.8,-0.2);
\draw[thick] (0.1,-0.2) -- (0.8,-0.4);

\end{tikzpicture}
\quad & & \\[-1.5cm]
 & \xrightarrow{\makebox{ {\small{substitution}} }}\ & \\[-1.3cm]
& &
\begin{tikzpicture}
\def \radius {1.5cm}

\draw[fill=blue](200:\radius) circle (2 pt) node [above] {}; 
\draw[fill=blue](270:\radius) circle (2 pt) node [above] {};
\draw[fill=blue](330:\radius) circle (2 pt) node [above] {};
   
  \node at (200:\radius + 0.4cm) {\tiny{$n-1$}} ;
  \node at (270:\radius + 0.2cm) {\tiny{$n$}} ;
  \node at (330:\radius + 0.2cm) {\tiny{$1$}} ;
 
  \draw[-, >=latex] ({180}:\radius) arc ({180}:{360}:\radius);

\draw[fill=black](-0.8,-0.2) circle (1.5 pt) node [above] {$i_1$};
\draw[fill=black](-0.3,-.8) circle (1.5 pt) node [right] {$i_{\text{new}}$};
\draw[fill=black](0.1,-0.2) circle (1.5 pt) node [above] {$i_2$};
\draw[fill=black](0.8,-0.4) circle (1.5 pt) node [above] {$i_3$};

\draw ({200}:\radius) -- (-0.8,-0.2);
\draw ({270}:\radius) -- (-0.3,-.8) ;
\draw ({330}:\radius) -- (0.8,-0.4);
\draw (-0.8,-0.2)  -- (-0.3,-.8);
\draw (-0.3,-.8) -- (0.1,-0.2);
\draw (0.1,-0.2) -- (0.8,-0.4) ;

\end{tikzpicture}
\end{array}
$$ 
In that case, the new characteristic path becomes
$$
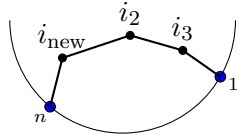
\begin{tikzpicture}
\def \radius {1.5cm}

\draw[fill=blue](230:\radius) circle (2 pt) node [above] {}; 
\draw[fill=blue](330:\radius) circle (2 pt) node [above] {};
   
  \node at (230:\radius + 0.2cm) {\tiny{$n$}} ;
  \node at (330:\radius + 0.2cm) {\tiny{$1$}} ;
 
  \draw[-, >=latex] ({180}:\radius) arc ({180}:{360}:\radius);

\draw[fill=black](-0.8,-0.5) circle (1.5 pt) node [above] {$i_{\text{new}}$};
\draw[fill=black](0.1,-0.2) circle (1.5 pt) node [above] {$i_2$};
\draw[fill=black](0.8,-0.4) circle (1.5 pt) node [above] {$i_3$};

\draw[thick] ({230}:\radius) -- (-0.8,-0.5);
\draw[thick] ({330}:\radius) -- (0.8,-0.4);
\draw[thick] (-0.8,-0.5) -- (0.1,-0.2);
\draw[thick] (0.1,-0.2) -- (0.8,-0.4) ;

\end{tikzpicture}
$$
\end{itemize}
\end{enumerate}

\medskip

With this set of diagrammatical operations in place, we define so-called Feynman rules, seen as a map $F$, which sends an admissible tree $T \in \mathcal{D}$ to a monomial of generating functions. We detail the generic case but the same process applies to word-decorated admissible trees.

Consider for instance the admissible $T \in \mathcal{D}$ with 4 leaves and two internal 1PI vertices of degree 3
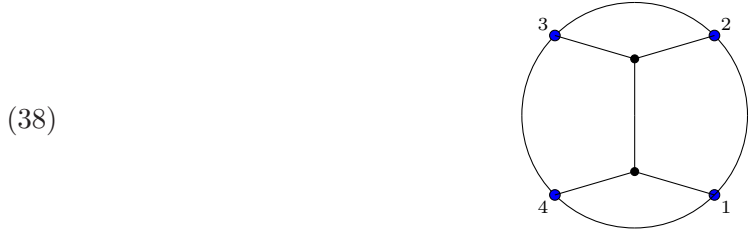
\begin{equation}
\label{FR: step0}
\begin{aligned}
\begin{tikzpicture}
\def \n {4}
\def \radius {1.5cm}
\foreach \s in {1,...,\n}
{
\draw[fill=blue]({360/\n * (\s - 1)- 45}:\radius) circle (2 pt) node [above] {};
  \node at ({360/\n * (\s + 3)- 45}:\radius + 0.22cm) {\tiny{$\s$}} ;
  \draw[-, >=latex] ({360/\n * (\s - 1)}:\radius) 
    arc ({360/\n * (\s - 1)}:{360/\n * (\s)}:\radius);
}

 %\draw[fill=black](0,0) circle (1.5 pt) node [above] {};

%\foreach \from/\to in {1/3}
%\draw ({360/\n * (\from - 1)}:\radius) -- ({360/\n * (\to - 1)}:\radius);

\foreach \from/\to in {1/3}
\draw ({360/\n * (\from - 1)- 45}:\radius) -- ({270}:\radius/2);

\draw ({90}:\radius/2) --  (0,0);

 \draw[fill=black]({90}:\radius/2) circle (1.5 pt) node [above] {};
 
\foreach \from/\to in {2/3}
\draw ({360/\n * (\from - 1)- 45}:\radius) --  ({90}:\radius/2) ;

\foreach \from/\to in {3/3}
\draw ({360/\n * (\from - 1)- 45}:\radius) -- ({90}:\radius/2) ;

\draw ({270}:\radius/2) --  (0,0);

\foreach \from/\to in {4/3}
\draw ({360/\n * (\from - 1)- 45}:\radius) -- ({270}:\radius/2);

\draw[fill=black]({270}:\radius/2) circle (1.5 pt) node [above] {};

\end{tikzpicture}
\end{aligned}
\end{equation}

Then we apply the following 3-steps procedure 

\begin{itemize}

\item Step 1: Assign (arbitrary, distinct) variables to all outgoing edges of the internal vertices. In the  example tree displayed in  \eqref{FR: step0}

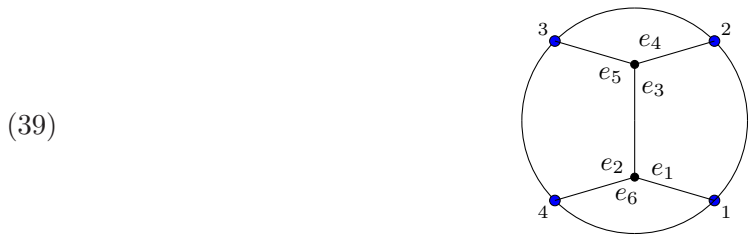
\begin{equation}
\label{FR: step1}
\begin{aligned}
\begin{tikzpicture}
\def \n {4}
\def \radius {1.5cm}
\foreach \s in {1,...,\n}
{
\draw[fill=blue]({360/\n * (\s - 1)- 45}:\radius) circle (2 pt) node [above] {};
  \node at ({360/\n * (\s + 3)- 45}:\radius + 0.22cm) {\tiny{$\s$}} ;
  \draw[-, >=latex] ({360/\n * (\s - 1)}:\radius) 
    arc ({360/\n * (\s - 1)}:{360/\n * (\s)}:\radius);
}

\foreach \from/\to in {1/3}
\draw ({360/\n * (\from - 1)- 45}:\radius) -- node [pos=.5,above] {\hspace{-0.3cm}\small{$e_1$}}  ({270}:\radius/2);

\draw ({90}:\radius/2) -- node [pos=.75,above] {\hspace{0.5cm}\small{$e_3$}}  (0,0);

 \draw[fill=black]({90}:\radius/2) circle (1.5 pt) node [above] {};
 
\foreach \from/\to in {2/3}
\draw ({360/\n * (\from - 1)- 45}:\radius) -- node [pos=.8,above] {\small{$e_4$}}   ({90}:\radius/2) ;
 
\foreach \from/\to in {3/3}
\draw ({360/\n * (\from - 1)- 45}:\radius) -- node [pos=.7,below] {\small{$e_5$}}  ({90}:\radius/2) ;

\draw ({270}:\radius/2) -- node [pos=.2,left] {\small{$e_2$}} (0,0);

\foreach \from/\to in {4/3}
\draw ({360/\n * (\from - 1)- 45}:\radius) -- node [pos=.95,below] {\hspace{-0.1cm}\small{$e_6$}} ({270}:\radius/2);

\draw[fill=black]({270}:\radius/2) circle (1.5 pt) node [above] {};

\end{tikzpicture}
\end{aligned}
\end{equation}

\item Step 2: 
\begin{itemize}

\item To each internal vertex of order $k$ on the characteristic path of $T$, assign a 1PI correlation functional ${\mathrm{L}}(\Phi)_{e_{i_1} \cdots e_{i_k}}$. E.g., to the order-3 vertex  

$$
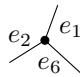
\begin{tikzpicture}
\def \radius {1.cm}
\draw[fill=black](250:\radius/2) circle (1.5 pt) node [above] {};
\draw ({230}:\radius) -- node [pos=.3,above] {\small{$e_2$}} (250:\radius/2);
\draw ({290}:\radius) -- node [pos=.3,left] {\small{$e_6$}}  (250:\radius/2);
\draw (0,0) -- node [pos=.7,right] {\small{$e_1$}}  (250:\radius/2);
\end{tikzpicture}
$$
in the characteristic path of \eqref{FR: step1} we assign the order-3 1PI vertex ${\mathrm{L}}(\Phi)_{e_1e_2e_6}$. Here, $e_1$ is on the edge joining the root (labelled 1) to the order-3 1PI vertex, and the indexing of ${\mathrm{L}}(\Phi)$ is done according to the anti-clockwise ordering of the labels of the vertex outgoing edges.

\item To each internal vertex of order $k$, which is not on the characteristic path of $T$, assign a 1PI correlation function ${\mathrm{L}}(0)_{e_{i_1} \cdots e_{i_k}} \in \mathbb{K}$. For instance, if $e_1$ is on the path from the root to the order-3 vertex\\
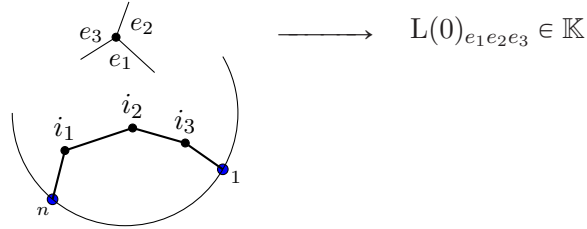
$$
\begin{array}{ccc}
\begin{tikzpicture}
\def \radius {1.cm}
\draw[fill=black](250:\radius/2) circle (1.5 pt) node [above] {};
\draw ({230}:\radius) -- node [pos=.3,above] {\small{$e_3$}} (250:\radius/2);
\draw ({290}:\radius) -- node [pos=.3,left] {\small{$e_1$}}  (250:\radius/2);
\draw (0,0) -- node [pos=.7,right] {\small{$e_2$}}  (250:\radius/2);
\end{tikzpicture}
\quad & & \\[-1cm]
 & \xrightarrow{\phantom{mmmm}} \ & \\[-0.5cm]
& &
\mathrm{L}(0)_{e_1e_2e_3} \in \mathbb{K}\\[0.1cm]
\begin{tikzpicture}
\def \radius {1.5cm}
\draw[fill=blue](230:\radius) circle (2 pt) node [above] {}; 
\draw[fill=blue](330:\radius) circle (2 pt) node [above] {};
   
  \node at (230:\radius + 0.2cm) {\tiny{$n$}} ;
  \node at (330:\radius + 0.2cm) {\tiny{$1$}} ;
 
  \draw[-, >=latex] ({180}:\radius) arc ({180}:{390}:\radius);

\draw[fill=black](-0.8,-0.5) circle (1.5 pt) node [above] {$i_{1}$};
\draw[fill=black](0.1,-0.2) circle (1.5 pt) node [above] {$i_2$};
\draw[fill=black](0.8,-0.4) circle (1.5 pt) node [above] {$i_3$};

\draw[thick] ({230}:\radius) -- (-0.8,-0.5);
\draw[thick] ({330}:\radius) -- (0.8,-0.4);
\draw[thick] (-0.8,-0.5) -- (0.1,-0.2);
\draw[thick] (0.1,-0.2) -- (0.8,-0.4) ;
\end{tikzpicture} & &
\end{array}
$$

\smallskip

\item
To each edge not on the characteristic path assign a $\mathbb{K}$-valued 2-cumulant, e.g.,
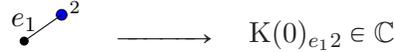
$$
\begin{array}{ccc}
\begin{tikzpicture}
\def \radius {1.5cm}
\draw[fill=blue](30:\radius) circle (2 pt) node [above] {};
\node at (30:\radius + 0.2cm) {\tiny{$2$}} ;
\draw[fill=black](0.8,0.4) circle (1.5 pt) node [above] {$e_1$};
\draw ({30}:\radius) -- (0.8,0.4);
\end{tikzpicture}
 & \xrightarrow{\phantom{mmmm}} \ & \\[-0.5cm]
& &
\mathrm{K}(0)_{e_12} \in \mathbb{C}\\[0.1cm]
\end{array}
$$
In general, the first lower index of the 2-cumulant should correspond to the dummy variable closest to the root.

\smallskip

\item
To each edge on the characteristic path assign the 2 point correlated correlation functional
$$
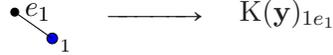
\begin{array}{ccc}
\begin{tikzpicture}
\def \radius {1.5cm}
\draw[fill=blue](330:\radius) circle (2 pt) node [above] {};
\node at (330:\radius + 0.2cm) {\tiny{$1$}} ;
\draw[fill=black](0.8,-0.4) circle (1.5 pt) node [right] {$e_1$};
\draw ({330}:\radius) -- (0.8,-0.4);
\end{tikzpicture}
 & & \\[-1cm]
 & \xrightarrow{\phantom{mmmm}} \ & \\[-.5cm]
& &
\mathrm{K}(\mathbf{y})_{1e_1} \\[0.1cm]
\end{array}
$$
In general, the first lower index of the 2 point functional should correspond to the label closest to the root on the graph.
\end{itemize}

\item Step 3: Take the product in such a way that the edges and vertices on the characteristic path appear from left to right. As the other terms are scalar-valued, they commute and can be arranged in any order.

\item Step 4: Multiply this product by $(-1)^p$, where $p$ is the number of internal vertices in the graph.
\end{itemize}

We let the reader put together our various arguments and constructions, and check that the Feynman rules provide a dictionary between 
\begin{itemize}
\item admissible graphs and components of the expansion of n point connected correlation functionals in terms of 2-point correlation functions and functionals and $k$-point 1PI correlation functions and functionals,
\item the action of the $K({\mathbf y})_{y_k}\frac{\partial}{\partial \phi_k}$ on products of connected and 1PI correlation functionals and the recursive generating process of admissible graphs we have described.
\end{itemize}
It follows that:

\begin{theorem}
\label{thm:treeexpansion}
Let $w= i_1\cdots i_n$ and $\mathcal{D}_w$ be the set of $w$-admissible trees. We denote the Feynman rules described above by the symbol $F: \mathcal{D}_w \to R_{\mathbf{X}}$ seen as a map sending decorated trees to non-commutative generating functions. Then
\begin{equation}
\label{Feynman}
	\mathrm{K}(\mathbf{y})_{y_{i_1} \cdots y_{i_n}}
	=\sum_{T \in \mathcal{D}_n} F(T).
\end{equation}
\end{theorem}

%%%%%%%%%%%%%%%%%%

\subsection{$0$-dimensional calculus}
\label{ssec:0dimcalc}

As a final remark, we consider the relations implied by the non-commutative Legendre transformation \eqref{Legendre5a} in the single-letter case $\mathbf{X}=\{x\}$. From the point of view of free probability, this corresponds to the univariate case --- the study of the distribution of a single non-commutative random variable by the properties of related generating series. In the physics literature, this is referred to as the 0-dimensional case. 

In this case generating series are commutative, which simplifies considerably the computations. The order-2 computation evaluated at $\mathbf{y}=0$ gives
\begin{equation}
\label{order2}
	\ell^{(2)} = (\mathrm{k}^{(2)})^{-1}.
\end{equation}
At higher orders we find
\begin{align}
	\mathrm{k}^{(3)} (\mathrm{k}^{(2)})^{-3} 
	&= -\ell^{(3)} 											\label{order3} \\
	\mathrm{k}^{(4)} (\mathrm{k}^{(2)})^{-4} 
	&= - \ell^{(4)} 
		+  2 \ell^{(3)}\mathrm{k}^{(2)} \ell^{(3)} 					\label{order4}\\
	\mathrm{k}^{(5)} (\mathrm{k}^{(2)})^{-5} 
	&= -\ell^{(5)} 
		-  5 \ell^{(4)}\mathrm{k}^{(2)} \ell^{(3)} 					\label{order5}\\		
	\mathrm{k}^{(6)} (\mathrm{k}^{(2)})^{-6} 
	&= - \ell^{(6)} 
		+ 6 \ell^{(5)}\mathrm{k}^{(2)} \ell^{(3)} 
		+ 6 \ell^{(4)}\mathrm{k}^{(2)}\ell^{(4)} 	
		- 14 \ell^{(3)}\mathrm{k}^{(2)}\ell^{(3)}\mathrm{k}^{(2)}\ell^{(3)} 	\label{order6}
\end{align}
From a physics point of view, these relations describe connected $n$-point Green's functions in terms of 1PI Green's functions (linked via propagators). Compare with the relations given in \cite[p.~43]{Brezin}. The precise nature of these relations from the viewpoint of free probability and non-crossing partitions will be discussed elsewhere. 

In general, it follows from the definition of admissible trees that the terms in the previous expansions can be obtained from the enumeration of primed Schr\"oder trees with $n$ leaves $PST(n)$. We get
\begin{proposition}
We have, for univariate free cumulants of order $n>2$:
\begin{equation}
	\mathrm{k}^{(n)} (\mathrm{k}^{(2)})^{-n}
	= \sum_{t \in PST(n-1)} (\mathrm{k}^{(2)})^{\text{e}(t)}\prod\limits_{v\in \text{vert}(t)}(-\ell^{(\text{ar}(v))}),
\end{equation}
where $\text{e}(t)$ is the number of edges connecting internal vertices in the tree, $\text{vert}(t)$ the set of internal vertices and $\text{ar}(v)$ the arity of the internal vertex $v$ .
\end{proposition}

%%%%%%%%%%%%%%%%%%%%%%%%%%%%%%%%
%%%%%%%%%%%%%%%%%%%%%%%%%%%%%%%%
%%%%%%%%%%%%%%%%%%%%%%%%%%%%%%%%
%%%%%%%%%%%%%%%%%%%%%%%%%%%%%%%%

\end{document}